\documentclass[journal]{IEEEtran}

\usepackage{afterpage}
\usepackage{algorithm}
\usepackage{algorithmic}
\usepackage{amscd}
\usepackage{amsfonts}
\usepackage{amsmath}
\usepackage{amssymb}
\usepackage{amsthm}
\usepackage{array}
\usepackage{balance}
\usepackage{booktabs}
\usepackage{footnote}
\usepackage{graphicx}
\graphicspath{{figures/}}
\usepackage{mathrsfs}
\usepackage{multirow}
\usepackage{pdflscape}
\usepackage{pgf-umlsd}
\usepackage{rotating}
\usepackage{stfloats}
\usepackage{subfigure}
\usepackage{tabularx}
\usepackage{wrapfig}
\usepackage{tikz}
\usetikzlibrary{shapes.geometric, arrows.meta, positioning, calc, fit, shapes.symbols, matrix, backgrounds}
\usetikzlibrary{arrows.meta, positioning}
\usetikzlibrary{calc}
\usepackage{float}
\usepackage{placeins}

\newcolumntype{C}[1]{>{\centering\arraybackslash}m{#1}}
\newcommand{\lit}[1]{\ifmmode\text{\footnotesize\textsf{#1}}\else{\footnotesize\textsf{#1}}\fi}
\newcommand{\chain}[1]{\ifmmode\text{\scriptsize\textsf{#1}}\else{\scriptsize\textsf{#1}}\fi}
\newcommand{\tablelit}[1]{\ifmmode\text{\textsf{#1}}\else{\textsf{#1}}\fi}
\newif\ifwhtbdfinal
\whtbdfinaltrue

\newif\ifrevfinal
\newcommand{\rev}[1]{#1}
\newenvironment{revblock}{}{}


\usepackage[compact]{titlesec}
\titlespacing*{\subsection}{0pt}{0.8ex plus 0.2ex minus 0.2ex}{0.4ex plus 0.1ex}

\theoremstyle{plain}
\newtheorem{lemma}{Lemma}

\newtheorem{proposition}[lemma]{Proposition}

\theoremstyle{definition}
\newtheorem{definition}{Definition}[section]

\theoremstyle{remark}

\usepackage{xr-hyper}
\usepackage[hidelinks]{hyperref}
\usepackage{xurl}



\newif\iffull
\fullfalse

\begin{document}%
\bstctlcite{ARMORbibcontrol}

\title{ARMOR: Adversary Cost Amplification via a Schema-Preserving Password-Hardening Retrofit for Resource-Constrained Systems}

\author{Avinash Srinivasan%
\thanks{This work has been submitted to the IEEE for possible publication. Copyright may be transferred without notice, after which this version may no longer be accessible.}%
\thanks{A. Srinivasan is with the Department of Cyber Science, United States Naval Academy, Annapolis, MD, USA (e-mail: srinivas@usna.edu).}%
\thanks{The views expressed in this article are those of the author(s) and do not reflect the official policy or position of the U.S. Naval Academy, Department of the Navy, the Department of War, or the U.S. Government.}}

\pagestyle{plain}

\maketitle

\begin{abstract}
Credential-database disclosure exposes Internet of Things gateways, edge servers and legacy verifiers to offline password guessing. Memory-hard key derivation functions (KDFs) are preferred but can exceed constrained-device budgets. ARMOR addresses this deployment gap with a schema-preserving verifier-side transform for database-only compromise, assuming an isolated Chain Registry. The registry holds secret per-account chains of injective operations; fixed terminal escape encoding preserves distinct pre-hash values under each chain while moving outputs outside printable ASCII. This disrupts evaluated stock tools' candidate coverage, but chain-aware attackers can generate the transformed bytes directly. We therefore evaluate a bounded chain-aware attack on $1{,}000$ paired salted-MD5 targets: it executes $89.35\times$ the plain arm's counted salted checks, reproduced on A6000 and A100. All target passwords are in the bounded pool; same-target ambiguity enumeration and simulated login trials recover every original password. These results describe the implemented attack, not minimum attacker cost. Bare-metal transforms take microseconds, while protected storage dominates the measured TrustZone-M components; these are component measurements, not complete authentication times. ARMOR thus provides bounded chain uncertainty under registry isolation with low transform cost, although an isolated high-entropy pepper remains cryptographically stronger. ARMOR composes as \lit{KDF(ARMOR(password))} when resources permit and the KDF accepts the complete transformed output.
\end{abstract}


\section{Introduction}\label{sec:intro}

Internet of Things (IoT) back ends and gateways increasingly maintain credential databases for authenticating devices and users. These stored verifiers are high-value breach targets because disclosure lets attackers guess passwords offline at their own speed. Recent large-scale credential leaks and breach-cost studies underscore the need for post-breach protections that survive stored-verifier disclosure~\cite{ibm2025costbreach,moab2024,megaleak2025,verizon2024dbir}.

This risk is pronounced for systems that still rely on fast hashes. Graphics processing unit (GPU) clusters have demonstrated hundreds of billions of NTLM guesses per second, and current password-storage guidance rejects fast general-purpose hashes such as \lit{MD5} and \lit{SHA-1}~\cite{Goodin2012Cluster,owaspPasswordStorageCheatSheet}. The standard remediation is to replace fast hashes with adaptive or memory-hard key derivation functions (KDFs), including \lit{bcrypt}, \lit{scrypt}, and \lit{Argon2}~\cite{provos1999future,percival2009stronger,biryukov2016argon2}. Yet KDF configurations that meaningfully slow GPU-accelerated offline attacks can exceed the memory, latency, or compatibility budgets of IoT, edge, embedded, and legacy systems. Section~\ref{sec:kdf_infeasibility} and Table~\ref{tab:iot_edge_platforms} quantify this constraint on the evaluated platforms. The constraint is especially important where gateways and constrained back ends authenticate large device and user populations~\cite{elhajj2019iotauthsurvey,trnka2022iotauthreview} while running on microcontroller- or edge-class hardware with little headroom for memory-hard KDFs~\cite{nist800-232,hatzivasilis2015lightphs,ometov2016cryptofeasibility}. ARMOR targets this deployment gap in environments where KDFs are the preferred remediation but cannot be provisioned at security-meaningful settings without breaking resource or compatibility constraints.

This paper presents ARMOR, a verifier-side password-hardening retrofit applied at the authentication server that stores and checks credential verifiers. Under \emph{Chain Registry} (CR) isolation, it adds a bounded per-account chain-search subproblem while preserving credential schemas in IoT and other resource-constrained systems. Building on the earlier STUMP scheme~\cite{srinivasan2015stump}, ARMOR composes $k$ lightweight byte- and string-level primitives from a public operation library $\mathcal{O}$. The assigned chain is stored in the CR, separated from the credential database (DB). \rev{The credential record remains unchanged because the CR uses the account identifier the verifier already stores.} The defender evaluates one assigned chain per login. A DB-only attacker can test other chains, but production constraints and output equivalence make the nominal chain count an upper bound rather than a work factor.

ARMOR relies on chain secrecy, not cryptographic hardness. Under Kerckhoffs' Principle~\cite{Petitcolas2025}, the algorithm and primitive library are public, and only the per-user chain is secret. If the CR is compromised, ARMOR's advantage collapses to the underlying hash or KDF. We therefore evaluate ARMOR for DB-only compromise, not as a replacement for KDFs, multi-factor authentication (MFA), or peppering. ARMOR targets offline guessing after verifier disclosure, complementing mechanisms for breach detection, attribution, and coercion response (\S\ref{sec:rw}).

Against a \textsc{naive} adversary that does not enumerate ARMOR chains (\S\ref{sec:adversary_model}), unhardened bit-diffusion chains disrupt the rule grammar for non-digit inputs ($\leq 0.10\%$ recovery) but leave a residual digit-only class, recovered up to $45.3\%$ by \lit{hashcat} rule attacks because per-byte rotation maps digits into printable ASCII. We characterize this limitation and close it with an injective escape-encoding primitive (\S\ref{sec:bounded_structural}).

An \textsc{aware} DB-only adversary knows $\mathcal{O}$ and depth $k$ but lacks the per-user chain (\S\ref{sec:adversary_model}). For this adversary, the hidden chain creates a finite conditional search problem governed by output equivalence (\S\ref{sec:results_aware}). The specified production construction draws its secret chain from the injective subset $\mathcal{O}_{\mathrm{inj}}$ and appends terminal \chain{ESC} to prevent transform-induced pre-hash equivalence and the digit-only printable bypass (\S\ref{sec:injective_escape}).

\begin{revblock}
STUMP~\cite{srinivasan2015stump}, ARMOR's direct predecessor, introduced pre-hash string manipulation as a lightweight mechanism for increasing offline-cracking cost, but selected transformations heuristically, without a formal adversary model, and demonstrated them only on a small password sample. ARMOR reworks this into a verifier-side retrofit with a keyed-chain construction, a nominal $|\mathcal{O}|^k$ chain-space model, and equivalence-aware conditional cost analysis. It adds \textsc{naive}/\textsc{aware} adversaries under Kerckhoffs' Principle, an injective escape-encoding primitive that closes the residual digit-only bypass, and large-scale evaluation against modern rule sets and constrained hardware.
\end{revblock}

ARMOR can replace \lit{Hash(password)} with \lit{Hash(ARMOR(password))}, or precede a standard KDF as
\lit{KDF(ARMOR(password))} when resources permit and the KDF accepts the complete transformed output (\S\ref{sec:armor_prehash}), without changing client interfaces, network protocols, or credential schemas for completed enrollments. Migration of existing accounts requires coordinated CR and verifier updates (\S\ref{sec:sys_flow}). Direct credential disclosure through phishing, malware, shoulder surfing, or coercion is out of scope.

The paper makes three contributions.

\begin{enumerate}
    \item \textbf{Verifier-side construction.} We specify a verifier-side retrofit with compact per-account chains held in a separate Chain Registry, bytewise transforms, and no required online auxiliary service. Completed enrollments add no credential-record fields. The production construction combines injective operations, a fixed terminal escape encoding, and enrollment bounds to preserve distinct pre-hash values for distinct passwords under each assigned chain and limit output expansion.

    \item \textbf{Systems evaluation.} We evaluate attacks on NVIDIA A100, RTX A6000 and RTX~PRO~6000 Blackwell GPUs and measure transform costs on a Linux edge platform and three microcontroller platforms. A bounded production-chain attack on A6000 and A100 reproduces the same counted salted work and account recovery ranks; same-target enumeration supports simulated login-valid recovery. The bare-metal measurements establish microsecond-scale transform minima. In a TrustZone-M prototype, separate measurements of protected storage, compute, and boundary components identify protected-storage access as the dominant measured secure-service cost.

    \item \textbf{Adversary analysis and validation.} Under DB-only compromise with the Chain Registry isolated, we distinguish \textsc{naive} attacks using unmodified tools from \textsc{aware} chain search. RockYou experiments characterize the digit-only bypass and evaluate escape encoding. An output-equivalence analysis quantifies expected output-class tests given the correct base password under a unit-cost model. Enrollment simulations demonstrate why marginally uniform chain assignment alone does not ensure unpredictability.
\end{enumerate}

Appendices A--H provide additional methodological details, tables, figures, and bounded-attack results. Table and figure numbering continues from the main text.

\section{Background and Related Work}\label{sec:rw}

The fundamental defense against password theft, dating back to Morris and Thompson's Unix \lit{crypt(3)} function, is to make offline guessing expensive through deliberately slow password hashing~\cite{morris1979password}. PKCS~\#5 v2.0 formalized this principle through \lit{PBKDF2} (RFC~2898~\cite{rfc2898}; formally analyzed by Yao and Yin~\cite{yao2005design}), and later KDFs such as \lit{bcrypt}, \lit{scrypt}, \lit{Argon2}, \lit{Balloon}, and \lit{Makwa} strengthened or adapted this line of defense~\cite{provos1999future,percival2009stronger,biryukov2016argon2,boneh2016balloon,pornin2015makwa,hatzivasilis2017phstatus}. Yet catastrophic server-side breaches continue, and aggregated mega-leaks such as the Mother of All Breaches (MOAB)~\cite{moab2024} and recent infostealer disclosures~\cite{megaleak2025} amplify credential-stuffing and password-reuse risk. These KDFs operate as the hash primitive itself, whereas ARMOR is orthogonal to it, hardening the input fed to whatever underlying hash a deployment uses, including lightweight constructions such as \lit{Ascon-Hash256} now standardized for constrained devices~\cite{nist800-232}.

Server-side defenses are also constrained by human password choice. Complex composition policies often backfire~\cite{komanduri2011passwords}, and large-scale studies show that password predictability remains exploitable~\cite{florencio2007large,bonneau2012science}. Probabilistic grammars, rule-based mangling, neural guessers, Monte-Carlo estimators, and learned rule sets all exploit this structure~\cite{Peslyak1996JTR,weir2009password,melicher2016fast,dellamico2015montecarlo,luo2026adrg}, while strength meters and guessing metrics attempt to reduce user burden~\cite{ur2016supporting,komanduri2016modeling}.

Human-chosen passwords follow skewed statistical distributions, with real-world passwords obeying Zipf-like laws across datasets and languages~\cite{wang2017zipf,hou2022zipf}. Small wordlists are therefore disproportionately effective, and structural attacks outperform naive dictionary enumeration, as Tatli demonstrated for rule-based cracking~\cite{tatli2015cracking}, motivating our rule-file evaluation.

GPU cracking illustrates the exposure of disclosed verifiers. In 2012, a 25-GPU cluster could test 350 billion NTLM hashes per second~\cite{Goodin2012Cluster}. Cracked credentials enable credential stuffing across services~\cite{nisenoff2023measuring}. Akamai recorded 193 billion credential-stuffing attacks during 2020~\cite{akamai2021gaming}, and its 2025 Defenders' Guide describes botnets automating such campaigns~\cite{akamai2025defenders}. These reports concern online credential reuse. ARMOR targets offline password guessing after verifier disclosure and complements breach alerting and credential-stuffing detection~\cite{thomas2019protecting,wang2020detecting}.

Related hardening mechanisms include MFA-bound credential hashing and password-hardening services. Nair and Song's \emph{Multi-Factor Credential Hashing Function} (MFCHF) derives resistance from an MFA witness and reports $20$--$160$ added bits of entropy~\cite[\S~7]{multifactorhashing2023}. Password-hardening services and password-hardened encryption (PHE), including \textit{Pythia}, partially oblivious commitments, \textit{Phoenix}, simple PHE, and threshold PHE, inject an external secret so that a leaked verifier is provably hard to crack offline~\cite{everspaugh2015pythia,schneider2016pommitments,lai2017phoenix,lai2018simple,brost2020threshold}. These stronger guarantees require an online auxiliary service, an additional login round trip, and sometimes distributed key management. ARMOR is not a cryptographic equivalent. It adds a local, per-user keyed transformation-chain search problem without an auxiliary service or changes to client interfaces, protocols, or completed credential records. The median conditional guessing number is approximately $52$ output-class tests at $k{=}3$ and $317$ at $k{=}5$. These unit-cost values assume the correct base password and do not measure the cost of jointly searching passwords and chains. ARMOR therefore complements embedded, legacy, intermittently connected, and resource-constrained deployments where MFA hardware, online hardening services, or memory-hard KDFs cannot be provisioned.

ARMOR also differs from conventional global peppering, which adds one shared secret to password hashing. ARMOR assigns a per-account chain. The historical 14-operation evaluation has $|\mathcal{O}|^k$ nominal sequences, while the production validator and output equivalence reduce the applicable conditional search cost (\S\ref{sec:results_aware}). Enclave-based variants such as \textit{SafeKeeper} place a keyed secret in a server-side trusted execution environment (Intel SGX) with a client-to-enclave attested channel~\cite{krawiecka2018safekeeper}. ARMOR's measured Pico~2W prototype instead keeps per-account chains in TrustZone-M secure storage (\S\ref{sec:results_secure_service}). This is a finite chain search problem under CR isolation, not the cryptographic protection of a high-entropy pepper.

Client-side deterministic generators such as \textit{PwdHash}, \textit{Hashpass}, \textit{LessPass}, and \textit{MasterPassword} mitigate password reuse without a stored vault~\cite{ross2005pwdhash,Chestnykh2010Hashpass,lesspass,masterpassword}, but a stolen site hash can enable offline recovery of the master password~\cite{llewellyn2017cracking}. ARMOR instead operates verifier-side after password submission and preserves client interfaces, protocols, and schemas.

Authentication for constrained and IoT devices is related but distinct because it targets device authentication and key exchange rather than hardening leaked human-password verifiers~\cite{elhajj2019iotauthsurvey,trnka2022iotauthreview,luo2021g2f}. Physical-unclonable-function protocols also bind identity to device hardware rather than to a stored verifier~\cite{almeer2023pufiot,gupta2023pufprotocol}. IoT gateways and edge back ends that store and verify credentials nonetheless present a distinct attack surface~\cite{lin2017iotsurvey,sha2020edgeiotsecurity}, and large-scale default-credential compromises such as the Mirai botnet underscore the stakes of protecting these stores~\cite{antonakakis2017mirai,enisa2020iotguidelines}.

Other post-breach defenses target adjacent problems: panic passwords address coercion~\cite{clark2008panic}, while honeytokens, honeywords, and \textit{Amnesia} pursue decoy-based breach detection and attribution~\cite{spitzner2003insider,juels2013honeywords,wang2021amnesia}. None raise the per-account cost of an offline rule-based attack on fast-hash credential stores. Tihanyi et al.\ show that hashes can be cracked without revealing the hash or plaintext~\cite{tihanyi2024privacy}, reinforcing the need for mechanisms that structurally degrade cracking efficacy rather than merely protect hash confidentiality.

\section{Foundations of Modern Password Protection}
\label{sec:foundations}

Server-side password storage combines \lit{salt}s, \lit{pepper}s, and KDFs for uniqueness, secrecy, and computational cost, respectively (Tables~\ref{tab:ARMOR_vs_all} and \ref{tab:ARMOR_vs_kdf_comprehensive}). ARMOR targets systems where constrained hardware cannot run KDFs at security-meaningful settings.

\begin{enumerate}
    \item \textbf{Salting} adds a unique random value before hashing, preventing precomputed attacks such as rainbow-table lookups. The \lit{salt} is stored with the hash. Its protection comes from uniqueness, not secrecy.

    \item \textbf{Peppering} adds a static secret before hashing, stored separately from the credential DB (e.g., firmware). DB-only compromise is insufficient without the \lit{pepper}, but \lit{pepper} recovery restores baseline attack surface.

    \item \textbf{KDFs}, including \lit{bcrypt}, \lit{scrypt}, and \lit{Argon2}, add configurable CPU and memory cost to resist parallel guessing.

\end{enumerate}

\subsection{When KDFs Cannot Be Deployed}
\label{sec:kdf_infeasibility}

We use \lit{Argon2id} as the representative memory-hard KDF, recommended by RFC~9106~\cite{rfc9106} and OWASP~\cite{owaspPasswordStorageCheatSheet}, with the standard $(m,t,p)$ parameters for \textit{memory cost}, \textit{time cost}, and \textit{parallelism}. On constrained IoT and edge hardware, memory-hard KDFs often exceed available memory, while iteration-hard KDFs can fit only at latency/security tradeoffs.

\textbf{1. Memory-hard KDFs exceed available memory.} We evaluate OWASP's recommended $(m,t,p)=(19\,\mathrm{MiB},2,1)$ profile~\cite{owaspPasswordStorageCheatSheet}. RFC~9106~\cite{rfc9106} specifies $(2\,\mathrm{GiB},1,4)$ as its default and $(64\,\mathrm{MiB},3,4)$ as its memory-constrained profile. The three evaluated microcontroller boards provide only $64$--$520$\,KiB total SRAM (Table~\ref{tab:iot_edge_platforms}). This $19$\,MiB profile exceeds those capacities by approximately $37$--$304\times$. On the application-class Cortex-A53, we measured $\approx 424$\,ms per authentication at this $(19\,\mathrm{MiB},2,1)$ profile and $\approx 1{,}445$\,ms at $m{=}64$\,MiB. OWASP also recommends lower-memory profiles at increased time cost.

\textbf{2. Iteration-hard KDFs fit but give limited GPU resistance.} \lit{bcrypt}~\cite{provos1999future} has a $\approx 4$\,KiB working set and runs on Cortex-M hardware, yet GPU and parallel-hardware implementations reach $\approx 10^{5}$ hashes per second per device~\cite{malvoni2014bcrypt,pervan2019distributed}, so low cost factors give little protection. Raising \lit{bcrypt}'s cost factor $c$ (the base-2 logarithm of the iteration count) is slow even on the Cortex-A53 ($\approx 310$\,ms at $c{=}10$, the OWASP minimum~\cite{owaspPasswordStorageCheatSheet}, and $\approx 1.24$\,s at $c{=}12$). These latencies conflict with concurrent IoT and edge workflows (multiple devices, machine-to-machine exchange, sub-second handover), leaving operators to run fast hashes and absorb post-breach risk.

\subsection{ARMOR: A Composable Pre-Hash Layer}
\label{sec:armor_prehash}

No single KDF profile spans the full embedded spectrum. ARMOR instead applies a deterministic, per-user keyed sequence to the hash function's input. Where a security-meaningful KDF cannot run, ARMOR can add post-breach protection alongside salting. If a KDF accepts the full transformed output, \lit{KDF(ARMOR(password))} retains its per-hash cost while \rev{adding a hidden-chain search subproblem}. This condition matters for \lit{bcrypt}: its 72-byte input limit is incompatible with the evaluated 128-byte password limit even when the chain has no length-doubling operation, because the terminal escape encoding, \chain{ESC} (defined in \S\ref{sec:injective_escape}), doubles the output (Eq.~\ref{eq:lenbound}). Direct \lit{bcrypt} composition would require a separately validated shorter input profile. Section~\ref{sec:ARMOR} defines the transformation hierarchy, keyed chains, and their \rev{$|\mathcal{O}|^k$ nominal space}.

\section{The ARMOR Framework: A Tiered Pre-Hash Transformation Model}\label{sec:ARMOR}
This section formalizes ARMOR's pre-hash transformation model. Unlike STUMP~\cite{srinivasan2015stump}, which selected transformations heuristically, ARMOR organizes primitives into a three-layer hierarchy and specifies the keyed-chain model that yields the nominal $|\mathcal{O}|^k$ adversary search space.

\subsection{ARMOR System Design Requirements}
\label{sec:design_requirements}
ARMOR is governed by five design requirements spanning \textit{schema transparency} (DR-1), \textit{functional integrity and extensibility} (DR-2), \textit{adversary cost amplification} (DR-3), \textit{separation of compromise domains} (DR-4), and \textit{resource compatibility} (DR-5). Table~\ref{tab:design_requirements} maps each requirement to its supporting analysis or evidence.

\subsection{The Three-Layer Transformation Hierarchy}
\label{sec:hierarchy}
ARMOR organizes the primitives in Table~\ref{tab:armor_primitives} into structural transforms, bit-level diffusion, and keyed primitive chains.

\subsubsection{Layer 1: Structural Character Transforms.} These primitives rewrite the password's character-level representation through case mapping, permutation, or expansion. They provide structural disruption and chain-space diversity. Examples include \chain{REV}, \chain{U-CASE}, \chain{L-CASE}, \chain{TOG}, and \chain{RFL}.

\subsubsection{Layer 2: Bitwise Diffusion.} Per-byte circular bit rotations (\chain{BSL}, \chain{BSR}) rotate each byte's eight bits left or right by one position, altering its ordinal value while preserving length.

\textbf{Caveat: digit-only bypass class.}
At the byte level, \chain{BSL} maps printable bytes \lit{0x20}--\lit{0x3f} back into printable ASCII (\lit{0x40}--\lit{0x7e}). Digits \lit{0x30}--\lit{0x39} therefore remain inside the grammar modeled by standard rule sets. We evaluate the digit-only subset ($4$--$19\%$ of length-filtered RockYou, depending on ranked-corpus window), with unhardened rule attacks recovering up to $45.3\%$ of it and at most $0.10\%$ of the evaluated non-digit targets. We suppress this bypass with high-bit \chain{BSL-X}/\chain{BSR-X} variants and close it in production with the injective escape primitive \chain{ESC} (\S\ref{sec:bounded_structural}).

\subsubsection{Layer 3: Keyed Primitive Chains.}
Layer 3 makes the \emph{selection and ordering} of primitives a per-user secret. From the public library $\mathcal{O}$ of $|\mathcal{O}|=14$ primitives (Table~\ref{tab:armor_primitives}) and fixed \emph{chain depth} $k$ (primary evaluated operating point $k=3$), the candidate chain space is the ordered $k$-fold product\vspace{-2mm}
\[
  \vspace{-2mm}
  \mathcal{O}^{k}=\{(o_1,\dots,o_k):o_i\in\mathcal{O},\ i=1,\dots,k\},\quad
  |\mathcal{O}^{k}|=|\mathcal{O}|^{k}.
\]
At $k=3$, the $14$-primitive evaluation library yields $14^{3}=2{,}744$ ordered chains. Let $P$ be the submitted password. A user's \emph{chain} $K\in\mathcal{O}^{k}$ is one such ordered sequence, and $T_K(P)$ denotes its output in chain order:
\vspace{-1mm}
\begingroup\small
\begin{equation}\label{eq:chain}
\vspace{-2mm}
  K=(op_1,\dots,op_k),\quad T_K(P)=(op_k\circ op_{k-1}\circ\cdots\circ op_1)(P).
\end{equation}
\endgroup

\rev{Each user's chain $K$ is held solely by the CR and is retrieved under the immutable internal account identifier \lit{user\_id} that the credential record already carries. ARMOR therefore adds no field to the credential DB. The stored tuple remains the pre-existing $\langle\lit{user\_id},\lit{hash},\lit{salt}\rangle$, and the CR alone holds the mapping $\lit{user\_id}\mapsto K$. Deployments must key the CR on a stable internal identifier, ordinarily the surrogate primary key, rather than on a mutable handle such as a username or email address, so that renaming an account cannot orphan its chain. Keying the CR by account rather than by an unlinkable random handle means that a CR-only exposure reveals which account holds which chain. Under the compromise model of \S\ref{sec:security_analysis}, an adversary holding the CR but not the verifier DB has no stored hashes to attack, and one holding both already has every chain. A DB-only compromise contains no direct account-to-chain mapping. Observations from an adversary's own chosen-password enrollments do not determine an independently assigned victim chain (\S\ref{sec:adversary_model}).} The adversary may evaluate up to $|\mathcal{O}|^{k}$ chains per password candidate. The historical chain-expanded evaluation enumerates this full library without validator pruning, whereas production chains are drawn from the validator-accepted set $\mathcal{V}_k$ (\S\ref{sec:chain_constraints}). The nominal $|\mathcal{O}|^{k}$ bound, evaluation-library chain equivalence, and production-set conditional guessing number $G_k$ are measured separately (\S\ref{sec:combinatorial_compounder}, \S\ref{sec:results_aware}).

\begin{table}[!hb]
\vspace{-3mm}
\renewcommand{\arraystretch}{1.15}
\caption{ARMOR Transformation Primitive Library}
\label{tab:armor_primitives}
\centering
\footnotesize
\setlength{\tabcolsep}{4pt}
\begin{tabular}{@{}c l l >{\raggedright\arraybackslash}p{1.75in}@{}}
\toprule
\textbf{Layer} & \textbf{Primitive} & \textbf{Category} & \textbf{Description} \\
\midrule
1 & \chain{L-CASE}  & Structural-1  & All characters to lowercase. \\
1 & \chain{U-CASE}  & Structural-1  & All characters to uppercase. \\
1 & \chain{CAP}     & Structural-1  & Uppercases first character, lowercases rest (ASCII letters). \\
1 & \chain{CAP-INV} & Structural-1  & Lowercases first character, uppercases rest (ASCII letters). \\
1 & \chain{TOG}     & Structural-1  & Inverts case of every character. \\
1 & \chain{REV}     & Structural-2  & Reverses the string. \\
1 & \chain{ROT-L}   & Structural-2  & Rotates characters left by one. \\
1 & \chain{ROT-R}   & Structural-2  & Rotates characters right by one. \\
1 & \chain{RFL}     & Structural-3  & Appends reversed string to itself. \\
1 & \chain{DUP}     & Structural-3  & Concatenates string with itself. \\
2 & \chain{BSL}     & Diffusion-1 & Per-byte circular shift left one bit. \\
2 & \chain{BSR}     & Diffusion-1 & Per-byte circular shift right one bit. \\
2 & \chain{BSL-X} & Diffusion-2  & \chain{BSL-X}: high-bit OR ($\lor\,\chain{0x80}$) after \chain{BSL}.$^{\ast}$ \\
2 & \chain{BSR-X} & Diffusion-2 & \chain{BSR-X}: high-bit OR ($\lor\,\chain{0x80}$) after \chain{BSR}.$^{\ast}$ \\
2 & \chain{ESC}     & Diffusion-3 & Escape-encode each byte to two non-printable bytes: $b\!\mapsto\!(\chain{0x80}\!\mid\! b{\gg}4,\ \chain{0x80}\!\mid\! b\,\&\,\chain{0xf})$. \\
\midrule
\multicolumn{4}{>{\raggedright\arraybackslash}p{\dimexpr\linewidth-2\tabcolsep\relax}}{\scriptsize \textit{Note:} The $14$-primitive evaluation library excludes \chain{ESC}. The non-injective case maps and starred high-bit variants \chain{BSL-X}/\chain{BSR-X} are evaluation-only. \chain{ESC} is the fixed terminal production escape (\S\ref{sec:injective_escape}).}\\
\bottomrule
\end{tabular}
\vspace{-2mm}
\end{table}

\subsection{Deployment Model}
\label{sec:deployment_model}
\label{sec:core_app}
ARMOR follows a split-knowledge deployment model aligned with credential-protection guidance, including NIST SP~800-57~\cite{nist800-57}, NIST SP~800-63B~\cite{nist800-63b}, and FIPS~140-3~\cite{fips140-3} when the CR is hardware backed. Per DR-4, the credential DB and CR reside in independent compromise domains.

ARMOR targets IoT gateways, edge back ends, legacy enterprise stacks with fixed verifier formats, and embedded deployments where memory-hard KDFs exceed resource budgets. Per DR-1, it inserts as a verifier-side application-layer pre-processing step, leaving verifier tuples and client protocols unchanged. The CR may reside in firmware, protected memory, a trusted execution environment (TEE) or TrustZone-protected region, a hardware security module (HSM), or an OS-level keystore. On embedded devices, this split can place the DB in data flash and the CR in firmware-protected instruction flash. \rev{The specified production construction composes primitives from the injective sublibrary $\mathcal{O}_{\mathrm{inj}}$ and then applies a fixed terminal \chain{ESC} to close the digit-only bypass (\S\ref{sec:injective_escape}). Depth is selected against the conditional guessing number $G_k$ of \S\ref{sec:results_aware} and the length and cost bound of \S\ref{sec:chain_constraints}. We evaluate depths $k{=}1$--$5$, using $k{=}3$ as the primary comparison point. Selecting a cost-optimal depth requires a deployment-specific assessment beyond the present evaluation.} Already-deployed hardware may provision the CR through an over-the-air (OTA) firmware update if it supports one and the prerequisites of \S\ref{sec:security_analysis} hold. Otherwise, provisioning requires factory setup, physical reflash, or scheduled service. Provisioning across multiple deployed devices was not evaluated.

\subsection{Defender-Side Registration and Authentication}
\label{sec:sys_flow}
ARMOR applies the identical keyed transformation at both registration and authentication. The flows differ only in the final step, storing the verifier at registration and comparing against it at authentication. Algorithm~\ref{alg:armor_defender} specifies both, and Fig.~\ref{fig:auth_flow} gives the authentication-flow diagram. During registration, \rev{the CR allocates a DR-2--valid chain $K$ and stores it under the user id $u$}. The verifier computes $P' = \textsc{armor}(P,K)=\chain{ESC}(T_K(P))$ from password $P$ (Eq.~\ref{eq:prod}). Indexed by \rev{$u$}, it stores the verifier tuple \rev{$\langle h,\, salt\rangle$}, where $h = H(P' \,\|\, salt)$ is the stored hash and $salt$ is a fresh per-user salt. The chain $K$ stays exclusively in the CR (DR-4)\rev{, and no ARMOR-specific field is added to the credential record}.

\begin{revblock}
\textbf{Deployment onto an existing store.} A retrofit must also handle accounts enrolled before ARMOR was deployed, whose stored $h$ was computed without a chain. On a successful authentication for such an account, the verifier holds $P$ in memory for that request. It can assign $K$ to $u$ and rewrite $h \gets H(\textsc{armor}(P,K)\,\|\,salt)$, retaining the existing salt or generating a fresh one. The plaintext need not be persisted, and a completed migration changes the same verifier fields as a routine password change. The CR entry and DB verifier are updated separately. A process or system failure between these writes can leave them inconsistent and prevent authentication. Reliable migration therefore requires coordinated updates and a recovery procedure, which this study does not specify or evaluate. A CR entry alone does not establish that migration completed or that an interrupted update can be safely repeated.

Two policy points bound the transition. Accounts that never authenticate do not migrate, so a deployment needs a cutover policy for dormant accounts, such as a required password reset. The legacy path must also have a deadline because accounts with no CR entry still use their original verifier before migration. After the deadline, a missing chain must be treated as an error. An account with no migration writes still has its original verifier, while an interrupted two-store update requires the recovery procedure noted above.
\end{revblock}

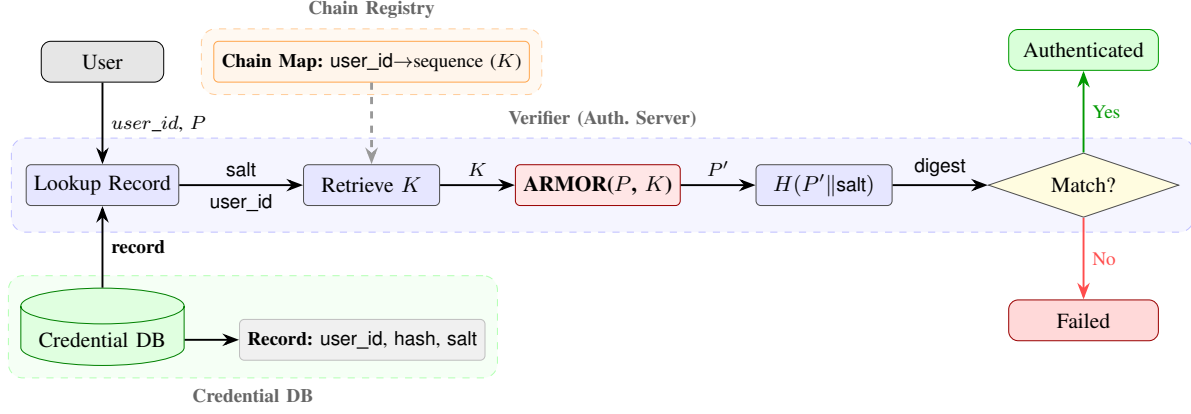
\begin{figure*}[!t]
\centering
\scalebox{0.9}{
\begin{tikzpicture}[
    startstop/.style={rectangle, rounded corners=4pt,
        minimum width=1.8cm, minimum height=0.6cm,
        text centered, draw=black, fill=gray!20, font=\small},
    process/.style={rectangle, rounded corners=2pt,
        minimum width=2.0cm, minimum height=0.6cm,
        text centered, draw=black!70, fill=blue!10, font=\small},
    armor/.style={rectangle, rounded corners=2pt,
        minimum width=2.0cm, minimum height=0.6cm,
        text centered, draw=red!60!black, fill=red!10,
        font=\small\bfseries},
    decision/.style={diamond, aspect=2.4,
        minimum width=2.8cm, minimum height=0.9cm,
        text centered, draw=black!70, fill=yellow!15, font=\small},
    outcome/.style={rectangle, rounded corners=4pt,
        minimum width=2.2cm, minimum height=0.6cm,
        text centered, draw=black!60, font=\small},
    chainbox/.style={rectangle, rounded corners=2pt,
        minimum width=3.0cm, minimum height=0.6cm,
        text centered, draw=orange!70, fill=orange!8, font=\footnotesize},
    db/.style={cylinder, shape border rotate=90,
        minimum width=2.4cm, minimum height=0.85cm,
        text centered, draw=green!50!black, fill=green!12,
        font=\small, aspect=0.25},
    recbox/.style={rectangle, rounded corners=2pt,
        minimum width=3.6cm, minimum height=0.6cm,
        text centered, draw=gray!50, fill=gray!12, font=\footnotesize},
    arrow/.style={-{Stealth[length=7pt]}, thick},
    dasharrow/.style={-{Stealth[length=7pt]}, very thick, dashed, gray!80},
    layerlabel/.style={font=\footnotesize\bfseries, black!60},
]

\node (lookup) [process]                          {Lookup Record};
\node (user)   [startstop, above=1.2cm of lookup] {User};
\node (chain)  [process,  right=1.82cm of lookup] {Retrieve $K$};
\node (armor)  [armor,    right=1.1cm of chain]   {ARMOR($P$, $K$)};
\node (hash)   [process,  right=1.1cm of armor]   {$H(P'\|\lit{salt})$};
\node (cmp)    [decision, right=1.4cm of hash]    {Match?};

\node (ok)     [outcome,  above=1.2cm of cmp,
                fill=green!15, draw=green!60!black] {Authenticated};

\node (fail)   [outcome,  below=1.2cm of cmp,
                fill=red!15, draw=red!60!black]    {Failed};

\node (db)     [db,     below=1.2cm of lookup]     {Credential DB};
\node (rec)    [recbox, right=0.8cm of db]
    {\textbf{Record:} \lit{user\_id}, \lit{hash}, \lit{salt}};

\node (cmap)   [chainbox, above=1.2cm of chain]
    {\textbf{Chain Map:} \rev{\lit{user\_id}}$\to$sequence ($K$)};

\draw[arrow] (user)   --
    node[right, font=\footnotesize, midway] {$user\_id$, $P$} (lookup);
\draw[arrow] (lookup) --
    node[above, font=\footnotesize, midway] {\lit{salt}}
    node[below, font=\footnotesize, midway] {\rev{\lit{user\_id}}} (chain);
\draw[arrow] (chain)  --
    node[above, font=\footnotesize, midway] {$K$} (armor);
\draw[arrow] (armor)  --
    node[above, font=\footnotesize, midway] {$P'$} (hash);
\draw[arrow] (hash)   --
    node[above, font=\footnotesize, midway] {\lit{digest}} (cmp);

\draw[arrow, green!60!black, thick]
    (cmp.north) -- node[right, font=\footnotesize] {Yes} (ok.south);
\draw[arrow, red!70, thick]
    (cmp.south) -- node[right, font=\footnotesize] {No}  (fail.north);

\draw[arrow] (db) --
    node[right, font=\footnotesize, midway] {\textbf{record}} (lookup);

\draw[arrow] (db.east) -- (rec.west);
\draw[dasharrow] (cmap) -- (chain);

\begin{scope}[on background layer]
\node[fill=orange!5, draw=orange!40, dashed, rounded corners=6pt,
          inner sep=5pt, fit=(cmap)] (cfgbox) {};
\node[layerlabel, above=1pt of cfgbox] {Chain Registry};
\end{scope}

\begin{scope}[on background layer]
\node[fill=blue!4, draw=blue!25, dashed, rounded corners=6pt,
          inner sep=6pt,
          fit=(lookup)(chain)(armor)(hash)(cmp)]
        (authbox) {};
\node[layerlabel, above=1pt of authbox] {Verifier (Auth.\ Server)};
\end{scope}

\begin{scope}[on background layer]
\node[fill=green!4, draw=green!30, dashed, rounded corners=6pt,
          inner sep=5pt, fit=(db)(rec)] (dbbox) {};
\node[layerlabel, below=1pt of dbbox] {Credential DB};
\end{scope}

\end{tikzpicture}
}
\caption{ARMOR authentication flow. The verifier retrieves the stored credential record from the database and keyed sequence $K$ from the CR (e.g., firmware, TEE, or HSM), transforms $P$ into $P'$, and compares the digest against the stored verifier.}
\label{fig:auth_flow}
\vspace{-4mm}
\end{figure*}

\begin{algorithm}[!t]
\caption{\small ARMOR Defender: Registration and Authentication.}
\label{alg:armor_defender}
\footnotesize
\begin{algorithmic}[1]
\REQUIRE public library $\mathcal{O}$ (production: $\mathcal{O}_{\mathrm{inj}}$) depth $k$;
\STATE \textbf{procedure} \textsc{Register}(\chain{user\_id}, password $P$)
\STATE \quad $K \gets \mathrm{CR.AssignValidChain}(\chain{user\_id}, \mathcal{O}_{\mathrm{inj}}, k)$ \\ \hspace{10mm}\COMMENT{\rev{stored in the CR under \chain{user\_id}}; $K$ is DR-2--valid (\S\ref{sec:chain_constraints})}
\STATE \quad \chain{salt} $\gets$ fresh random \chain{salt}
\STATE \quad $\mathrm{DB}[\chain{user\_id}] \gets \langle\, H(\textsc{armor}(P,K)\,\|\,\chain{salt}),\; \chain{salt}\,\rangle$
\STATE \textbf{procedure} \textsc{Authenticate}(\chain{user\_id}, password $P$)
\STATE \quad $\langle h,\, \chain{salt}\rangle \gets \mathrm{DB}[\chain{user\_id}]$
\STATE \quad $K \gets \mathrm{CR}[\chain{user\_id}]$
\STATE \quad \textbf{return} $H(\textsc{armor}(P,K)\,\|\,\chain{salt}) = h$
\end{algorithmic}
\end{algorithm}

\subsection{Chain Construction Constraints}
\label{sec:chain_constraints}

At enrollment, ARMOR rejects chains containing \emph{adjacent} operations that cancel. The constrained-inverse set $\mathcal{C}$ contains \lit{(BSL, BSR)} and \lit{(ROT-L, ROT-R)}, so $|\mathcal{C}|=2$. The validator also rejects adjacent repetitions of the self-inverse operations \chain{REV} and \chain{TOG}, namely \chain{REV}\,\chain{REV} and \chain{TOG}\,\chain{TOG}.

These local checks do not eliminate every algebraic equivalence. Non-adjacent combinations can still collapse (e.g.\ \chain{REV}\,\chain{ROT-L}\,\chain{REV} equals one \chain{ROT-R}), and distinct accepted chains can share outputs. The validator accepts the non-adjacent degeneration \lit{T4-BC3} examined in \S\ref{sec:t4_anomaly}. The conditional guessing number $G_k$ (\S\ref{sec:results_aware}, Table~\ref{tab:guessing_number}) accounts for these equivalences by grouping accepted chains with identical outputs.

The production set $\mathcal{V}_k\subseteq\mathcal{O}_{\mathrm{inj}}^{\,k}$ contains the chains from the production library of Eq.~(\ref{eq:Oinj}) that satisfy both the adjacency rules above and the enrollment length bound $e(K)\le e_{\max}=3$ of Eq.~(\ref{eq:lenbound}). Here $e(K)$ counts the length-doubling \chain{DUP} and \chain{RFL} steps. Every reported $|\mathcal{V}_k|$ is the size of this jointly constrained set. The expansion-count limit removes additional chains only at $k\ge4$. Section~\ref{sec:operational_guarantees} defines $G_k$ and its bounds for this production set.

Both checks take $O(k)$ at enrollment and determine which chains the CR may assign. The Trusted Firmware-M (TF-M) prototype also revalidates each retrieved chain during authentication. Table~\ref{tab:secure_service} excludes that validation work from its timed components. Each new primitive registered under DR-2 must declare inverse relationships before deployment.

\begin{revblock}
\textbf{Length safety and bounded work.} \chain{DUP} and \chain{RFL} each double their input. Production then applies the fixed terminal escape \chain{ESC} (defined in \S\ref{sec:injective_escape}), which doubles the length again. Writing $e(K)$ for the number of length-doubling operations in $K$, a password of $L_{\mathrm{pw}}$ bytes produces
\vspace{-1mm}
\begin{equation}\label{eq:lenbound}
\vspace{-1mm}
L_{\mathrm{out}} \;=\; L_{\mathrm{pw}}\cdot 2^{\,e(K)+1} \;\le\; \min(B,\,L_{\mathrm{KDF}}),
\end{equation}
Here the $+1$ accounts for the fixed terminal \chain{ESC}, $B$ is the transformed-value capacity, and $L_{\mathrm{KDF}}$ is the downstream input limit. A KDF that truncates its input would discard assigned transform output. Production ARMOR sets $P_{\max}=128$ B and $e_{\max}=3$, giving $B_{\min}=P_{\max}2^{\,e_{\max}+1}=2048$ B. The production path rejects inputs longer than 128 bytes before transformation. The CR checks Eq.~(\ref{eq:lenbound}) at enrollment in $O(k)$, alongside the adjacency rules. A chain that would overflow $B$ for any admissible password is invalid and is never assigned. These bounds prevent unbounded work from transform expansion on a request. No operation is skipped at run time, since that would evaluate a different chain and reject the account's own legitimate password.
\end{revblock}
\subsection{Injective Universal Escape}
\label{sec:injective_escape}
The terminal production hardening primitive must satisfy two properties. First, it must be injective, so distinct plaintexts under the same fixed chain have distinct ARMOR pre-hash values (\S\ref{sec:auth_equivalence}). Second, it must move outputs outside the printable-ASCII grammar that cracking rule sets exploit (\S\ref{sec:bounded_structural}). The high-bit hardened primitives each satisfy only one property, since high-bit OR is escaping but non-injective, while high-bit XOR is injective but not universally escaping. The escape primitive \chain{ESC} resolves this tradeoff through length-expanding, non-printable encoding (Table~\ref{tab:armor_primitives}; Appendix~\ref{app:II}). \rev{The specified production construction therefore fixes \chain{ESC} as a \emph{terminal} escape applied to every transformed password, and draws the secret chain from the eight non-terminal injective operations}
\vspace{-1mm}
\begin{equation}\label{eq:Oinj}
\vspace{-1mm}
\mathcal{O}_{\mathrm{inj}}=\{\chain{REV},\chain{ROT-L},\chain{ROT-R},\chain{RFL},\chain{DUP},\chain{TOG},\chain{BSL},\chain{BSR}\},
\end{equation}
\rev{with $|\mathcal{O}_{\mathrm{inj}}|=8$, so the production construction is}
\vspace{-1mm}
\begin{equation}\label{eq:prod}
\vspace{-1mm}
P' \;=\; \chain{ESC}\bigl(T_K(P)\bigr), \qquad K \in \mathcal{V}_k \subseteq \mathcal{O}_{\mathrm{inj}}^{\,k},
\end{equation}
\rev{where $T_K$ applies the secret chain of Eq.~(\ref{eq:chain}) and \chain{ESC} is fixed rather than selected. The secret is $K$ alone: \chain{ESC} is public and always present, so it contributes hardening but no chain-space entropy. This is the construction all production claims refer to, and the one measured in \S\ref{sec:results_aware}.} Every such chain is injective by construction, and the terminal escape closes the digit-only bypass without changing stored verifier formats. Case maps and high-bit \chain{BSL-X}/\chain{BSR-X} variants remain only in the $|\mathcal{O}|=14$ evaluation library that defines the nominal chain space (\S\ref{sec:combinatorial_compounder}).

\section{ARMOR Security Analysis and Threat Model}
\label{sec:security_analysis}
This section formalizes the verifier-side security boundary for IoT, edge, and legacy deployments assumed throughout the evaluation. ARMOR addresses DB-only offline compromise, where the adversary obtains stored verifiers but not the CR. Table~\ref{tab:threat_scope} delineates threat scope, including CR compromise, which reduces security to the underlying verifier defense.

\subsection{Threat Classification and Scope}
Password threats include (1) online guessing against an authentication service, including targeted guesses based on personal information~\cite{wang2016targeted}; (2) offline guessing against disclosed verifiers; and (3) direct acquisition of plaintext passwords. ARMOR's in-scope asset is plaintext-password confidentiality after verifier disclosure, with account integrity and stored-verifier integrity treated as secondary properties. \S\ref{sec:adversary_model} formalizes the DB-only adversaries.

\subsection{Adversary Models}
\label{sec:adversary_model}
We model two DB-only offline cracking adversaries with GPU-class hardware and a separate chosen-password enrollment adversary. The cracking adversaries operate on the exfiltrated verifier DB without the CR mapping \rev{$\lit{user\_id}\mapsto K$}; under Kerckhoffs' Principle~\cite{Petitcolas2025}, the ARMOR construction (algorithm, library $\mathcal{O}$, depth $k$) is public. The two cracking adversaries differ in whether they enumerate transformation chains.

\textbf{Adversary 1: \textsc{naive} (non-chain-enumerating).} This adversary applies wordlists, brute-force masks, and public rule files that encode human password-mutation patterns~\cite{tatli2015cracking,weir2009password}, but does not enumerate ARMOR chains.

\textbf{Adversary 2: \textsc{aware}.} This adversary knows $\mathcal{O}$ and $k$ but not the per-user chain $K$. For each candidate password, it can enumerate up to $|\mathcal{O}|^k$ ordered chains. Each transformed candidate can be reused across target accounts, but each account requires a separate salted verifier check. Section~\ref{sec:operational_guarantees} distinguishes nominal enumeration from conditional chain cost.

\begin{algorithm}[t]
\caption{\small Historical 14-operation \textsc{aware} evaluation with nominal $|\mathcal{O}|^k$ chain expansion. Production appends terminal \chain{ESC} and uses $\mathcal{V}_k$.}
\label{alg:aware_adversary}
\footnotesize
\begin{algorithmic}[1]
\REQUIRE credential records $\langle \lit{user\_id},\, h,\, \lit{salt}\rangle$; wordlist $\mathcal{W}$; public $\mathcal{O}$, depth $k$
\ENSURE matching (account, candidate, chain) triples
\FORALL{candidate $P \in \mathcal{W}$}
  \FORALL{ordered chain $c=(op_1,\dots,op_k) \in \mathcal{O}^k$}
    \STATE $P' \gets (op_k\circ\cdots\circ op_1)(P)$ \COMMENT{chain hypothesis reusable across accounts}
    \FORALL{exfiltrated credential record $\langle \lit{user\_id},\, h,\, \lit{salt}\rangle$}
      \IF{$H(P'\,\|\,\lit{salt}) = h$}
        \STATE record the matching triple $(\lit{user\_id},P,c)$
      \ENDIF
    \ENDFOR
  \ENDFOR
\ENDFOR
\end{algorithmic}
\end{algorithm}

A matching triple need not identify the original password or a password accepted under the account's assigned chain. Production injectivity holds for each fixed chain, not across different chains. Algorithm~\ref{alg:aware_adversary} records matches without resolving this ambiguity.

\begin{revblock}
\textbf{Adversary 3: chosen-password enrollment.} An adversary may open many accounts, choose their passwords, and use a breached DB to test those known passwords against every chain in $\mathcal{V}_k$. This can recover the adversary's own chains up to output equivalence. Let $K_v$ be an independently assigned victim chain, and let $E_A$ contain only observations from the adversary's own enrollments. If assignment is independent and uniform over $\mathcal{V}_k$, these observations do not change the victim-chain distribution. For $c\in\mathcal{V}_k$,
\[
\Pr[K_v=c\mid E_A]=\Pr[K_v=c]=1/|\mathcal{V}_k|.
\]
This equality does not describe the entire post-breach view, which may also include the victim's verifier.

The enrollment tests in Appendix~\ref{app:VI} evaluate marginal uniformity (T1), serial dependence (T2), identifier dependence (T3), victim-chain prediction (T4), and output-equivalence structure (T5). The evaluation uses up to $10^6$ enrollments per scheme per depth and adversaries holding up to $10^5$ accounts. The specified cryptographically secure pseudorandom number generator (CSPRNG) sampler produced no rejection in T1--T3 and no three-standard-error flag in T4. T5 is descriptive, not a leakage test. Round-robin allocation and derivation from account identifiers are each exactly uniform over $\mathcal{V}_k$, so marginal uniformity alone detects neither dependent scheme. Both disclose the victim chain under the tested model, giving top-$1$ prediction accuracy $1.000$ and up to $21{,}822\times$ chance at $k{=}5$. Assignment therefore needs independence and non-derivability as well as a suitable marginal distribution.
\end{revblock}

\subsection{Separation of Compromise Domains}
\label{sec:security_privilege}
ARMOR's operational cost amplification depends on DR-4. The credential DB and CR must remain separate compromise domains, so DB exfiltration alone does not expose the account-to-chain mapping, although the operation library is public. This is an instance of separation of privilege~\cite{saltzer1975protection}. ARMOR realizes this separation at two assurance levels, depending on platform capability.

\textbf{1. Hardware-Enforced Trust Boundaries.}
Modern secure IoT microcontrollers (e.g., Cortex-M33) provide hardware-backed isolation (TrustZone-M~\cite{armTrustZoneCM}, execute-only memory (XOM)~\cite{kwon2019uxom}) but are SRAM-constrained ($\le 520$\,KiB), offering trust boundaries without enough memory for recommended memory-hard KDFs (\S\ref{sec:kdf_infeasibility}). ARMOR can place the compact per-account chain state inside the isolated region. Its transform buffers are larger and separately bounded by Eq.~(\ref{eq:lenbound}).

\textbf{2. Firmware-Isolation Fallback.}
Without TEE support, DR-4 relies on the firmware-flash split discussed in \S\ref{sec:deployment_model}~\cite{stm32l4protections}. This fallback is weaker than silicon-enforced isolation. It raises the required compromise from DB exfiltration to firmware extraction but does not protect the CR under full firmware readout~\cite{obermaier2017firmware}. Where a KDF is deployed and accepts the complete transformed output, CR exposure removes ARMOR's chain-search benefit but leaves the KDF layer in place. In hash-only deployments, the remaining protection comes from the underlying salted hash and any independently protected mechanism.

\subsection{Authentication Equivalence}
\label{sec:auth_equivalence}
Because ARMOR is applied before hashing, $H(\textsc{armor}(P,K)\,\|\,\lit{salt})$, a non-injective transform can make distinct plaintexts $P_1\neq P_2$ with $\textsc{armor}(P_1,K)=\textsc{armor}(P_2,K)$ authenticate to the same account. This affects authentication integrity, not just attacker cost. DR-2 therefore requires production chains drawn from $\mathcal{O}_{\mathrm{inj}}$ (\S\ref{sec:injective_escape}). The non-injective evaluation primitives quantify the risk excluded by this restriction.

Over $1$M sampled RockYou entries ($999{,}997$ distinct plaintexts), the most collision-prone non-injective chain yields $3.99\%$ fewer distinct transformed values than distinct plaintexts, versus $\le 0.08\%$ for near-injective chains. This distinct-value loss measures within-chain output collapse, not the probability that an arbitrary wrong guess authenticates. For a fixed account and chain $K$, false-accept exposure is the attacker's guess-distribution mass on the equivalence class $[P]_K\setminus\{P\}$, not a uniform random-string probability.

\begin{lemma}[Fixed-chain production injectivity]
\label{lem:production_injectivity}
For every fixed accepted production chain $K\in\mathcal{V}_k$, the transform $P\mapsto\textnormal{\textsc{armor}}(P,K)=\chain{ESC}(T_K(P))$ is injective on its admissible input domain. For distinct passwords under that chain and the same salt, equal stored verifiers would require a collision of the underlying hash $H$.
\end{lemma}
\begin{proof}
Each operation in $\mathcal{O}_{\mathrm{inj}}$ is injective, so their composition $T_K$ is injective. Terminal \chain{ESC} preserves injectivity because its pair of output nibbles uniquely determines each input byte. Therefore $\chain{ESC}\circ T_K$ is injective.
\end{proof}
The result is for a fixed chain and does not establish uniqueness across different chains.

\subsection{ARMOR versus Per-User Peppering}
\label{sec:peppering}
Per-user peppering is cryptographically stronger when protected per-user secrets can be provisioned and managed correctly. ARMOR offers compact, schema-preserving preprocessing and a finite per-account chain search problem rather than high-entropy secret protection. Exposure of a pepper or CR removes the corresponding secret-dependent benefit. Once the assigned chain is known, an attacker can apply its transformation before hashing. Failure of an unmodified tool to generate that output does not protect against this attacker, who differs from the DB-only \textsc{aware} adversary. The mechanisms can be layered as \lit{Hash}(\lit{pepper}$\|$\lit{ARMOR(password)}) or, when resources permit and the KDF accepts the full combined input, \lit{KDF}(\lit{pepper}$\|$\lit{ARMOR(password)}). Appendix~\ref{app:IV} compares storage and hash-based message authentication code (HMAC) overhead.

We claim no general latency advantage over peppering. For the measured secure-service components (\S\ref{sec:results_secure_service}), the per-user-pepper sum is $1.42\times$ ARMOR's when each includes its per-account storage read. A global pepper with its key held live has the lower sum. The designs also manage different secret state. Disclosure of a global pepper affects every account using it, whereas disclosure of one per-user pepper affects that account. ARMOR stores less protected per-account state than a $128$-bit per-user pepper and gives a bounded, conditional chain cost rather than the cryptographic protection of a high-entropy secret. Rotating an assigned chain requires re-enrollment because the stored verifier depends on the transformed password, as it does when changing a per-user pepper.

Among post-breach defenses that rely on secret material stored separately from the credential DB, ARMOR combines compact per-account chain state with low transform cost and no required online auxiliary service (Table~\ref{tab:design_space}). A protected high-entropy pepper or password-hardening service provides stronger cryptographic protection under DB-only compromise. ARMOR instead targets constrained deployments through bytewise preprocessing without protocol changes. Its per-account state still requires protected storage, independent assignment, migration, rotation, and recovery procedures. Per-account isolation is also available with per-user peppers and is not unique to ARMOR.

\begin{table*}[!t]
\renewcommand{\arraystretch}{1.25}
\caption{Design space of secret-based post-breach password defenses. A protected pepper or hardening service provides stronger cryptographic protection under DB-only compromise. The mechanisms also differ on the remaining axes.}
\label{tab:design_space}
\centering
\footnotesize
\setlength{\tabcolsep}{4pt}
\begin{tabularx}{\textwidth}{@{}l >{\raggedright\arraybackslash}X >{\raggedright\arraybackslash}X >{\raggedright\arraybackslash}X c l l@{}}
\toprule
\textbf{Mechanism} & \textbf{Protected secret} & \textbf{DB-only breach} & \textbf{If secret also leaks} & \textbf{Aux.\ service} & \textbf{Granularity} & \textbf{Constrained-viable} \\
\midrule
Peppering (co-located)          & Global $128$-bit key            & Resists (weak boundary)   & Baseline restored       & No               & System-wide  & Yes \\
Isolated pepper (TEE/HSM)        & Global key in enclave           & Resists (cryptographic)   & Baseline restored       & No               & System-wide  & Needs TEE/HSM \\
Hardening service (PHE)          & Server key + protocol           & Resists (provable)        & Baseline restored       & Yes (round trip) & System-wide  & No \\
ARMOR                            & Per-account chain ($422$ accepted choices at $k{=}3$) & Finite work-factor & \rev{Baseline restored} & No           & Per-account  & Yes \\
\bottomrule
\addlinespace[2pt]
\multicolumn{7}{@{}>{\raggedright\arraybackslash}p{\textwidth}@{}}{\scriptsize \textit{Note:} ``DB-only breach'' means that only the credential DB is exfiltrated. If the protected pepper, service key, or account chain also leaks, its keyed benefit is lost for the affected accounts. An attacker who knows the chain can apply it before hashing, even if stock tools cannot (\S\ref{sec:bounded_structural}). Isolated pepper and PHE offer cryptographically stronger DB-only protection than ARMOR. The mechanisms differ in secret granularity, key management, auxiliary-service dependence, and constrained-device viability.}\\
\end{tabularx}
\end{table*}

\subsection{Operational Cost Model and Resistance Mechanisms}
\label{sec:operational_guarantees}
\label{sec:combinatorial_compounder}
Given the \textsc{aware} model and DR-4, ARMOR expands each plaintext guess into a chain-indexed candidate set per account (Algorithm~\ref{alg:aware_adversary}). The historical $|\mathcal{O}|{=}14$ evaluation library enumerates the complete nominal space ($2{,}744$ chains at $k{=}3$ and $38{,}416$ at $k{=}4$), but chain equivalence means these counts are not work factors. The production construction is measured separately by its conditional unit-cost guessing number $G_k$ over $\mathcal{O}_{\mathrm{inj}}$ (\S\ref{sec:results_aware}). Production draws $K$ from the eight non-terminal operations of Eq.~(\ref{eq:Oinj}), applies the adjacent-pair and length checks, and then applies the fixed terminal escape of Eq.~(\ref{eq:prod}). The $8^3{=}512$ nominal orderings at $k{=}3$ become $|\mathcal{V}_3|{=}422$ accepted production chains (Table~\ref{tab:guessing_number}). Larger depths increase the accepted set ($|\mathcal{V}_4|{=}3{,}054$, $|\mathcal{V}_5|{=}21{,}822$), while the measured secure-service stress case at $k{=}5$ remains small relative to its protected-storage path (\S\ref{sec:results_secure_service}). In attacks against multiple accounts, a transformed candidate can be reused, but each account requires a separate salted verifier check. The conditional guessing number therefore does not establish an end-to-end cracking slowdown.

\begin{definition}[Conditional guessing number]
\label{def:conditional_guessing}
Fix the correct base plaintext $p$ and depth $k$. Partition the accepted chains $K\in\mathcal{V}_k$ by $T_K(p)$ into $E=E^{\mathrm{prod}}_k(p)$ output classes with sizes $m_1\ge\cdots\ge m_E>0$ and $\sum_i m_i=|\mathcal{V}_k|$. Under uniform hidden-chain assignment, an adversary given $p$ tests these classes in nonincreasing size order. The expected number of unit-cost class tests is
\[
G_k(p)=\frac{1}{|\mathcal{V}_k|}\sum_{i=1}^{E}i\,m_i .
\]
The fixed, injective terminal \chain{ESC} preserves these classes.
\end{definition}

\begin{proposition}[Bounds]
\label{prop:gk}
For every base plaintext $p$ and depth $k$,
\[
1\le G_k(p)\le E^{\mathrm{prod}}_k(p)\le|\mathcal{V}_k|
\le|\mathcal{O}_{\mathrm{inj}}|^k .
\]
Hence the production distinct-output count $E^{\mathrm{prod}}_k(p)$ and the nominal production chain count $|\mathcal{O}_{\mathrm{inj}}|^{k}$ are both upper bounds on $G_k(p)$. Reporting either in place of $G_k$ overstates the conditional chain subproblem.
\end{proposition}

\begin{proof}
Every class index lies in $\{1,\dots,E^{\mathrm{prod}}_k(p)\}$, so $G_k(p)=\sum_i (m_i/|\mathcal{V}_k|)\,i\in[1,E^{\mathrm{prod}}_k(p)]$. Each of the $E^{\mathrm{prod}}_k(p)$ classes is nonempty, so $E^{\mathrm{prod}}_k(p)\le\sum_i m_i=|\mathcal{V}_k|$; and the enrollment validator only removes chains from $\mathcal{O}_{\mathrm{inj}}^{k}$, so $|\mathcal{V}_k|\le|\mathcal{O}_{\mathrm{inj}}|^{k}$.
\end{proof}

For a nonuniform password distribution, success is the cumulative probability mass of the guesses tried. \rev{The conditional $G_k(p)$ is not a password-space multiplier. A joint attacker must also order unknown plaintexts and evaluate the salted verifiers. The bounded production experiment in \S\ref{sec:results_aware} measures salted checks and CUDA search time for one fixed joint schedule; it does not establish the minimum cost of an optimized end-to-end attacker or justify multiplying separate component measurements.} Because chains are assigned independently, discovering one chain does not identify another.

\textbf{Chain diversity versus character entropy.} This conditional chain cost is a property of the authentication system, not an increase in per-character entropy. It arises from uncertainty about the assigned chain. Under each fixed production chain, injectivity preserves distinct pre-hash values for distinct passwords (\S\ref{sec:injective_escape}).

\textbf{A finite, composable chain subproblem.} \rev{$G_k$ is a bounded conditional, unit-cost
quantity, not an end-to-end cracking multiplier or a cryptographic-hardness assumption. Its median is
$52\times$ at the primary evaluated operating point $k{=}3$, rising to $317\times$ at $k{=}5$}
(Table~\ref{tab:guessing_number}). We report it instead of the validator-accepted chain count because the two
differ by $8.1\times$ at the primary evaluated operating point ($k{=}3$). Where a memory-hard KDF can be deployed, ARMOR transforms the submitted password before the KDF evaluates it,
$\lit{KDF}(\lit{ARMOR}(P))$. The complete offline cost
must be established over the joint plaintext, chain, salt, and verifier-check schedule.

The following subsection explains rule-pipeline disruption against the \textsc{naive} adversary. Neither candidate-coverage limitations nor conditional $G_k$ alone establish complete offline attack cost.

\subsection{Structural Disruption of Rule-Based Pipelines}
\label{sec:structural_invalidation}
Standard rule sets model human mutation patterns such as case toggles, substitutions, and keyboard walks. Layer-$2$ bit-diffusion primitives (\chain{BSL}, \chain{BSR}, Table~\ref{tab:armor_primitives} and \S\ref{sec:hierarchy}) can produce outputs outside the evaluated rule grammars. A \textsc{naive} attacker may therefore include the correct plaintext in its candidate set yet fail to generate the required pre-hash value.

The digit-only bypass and measured hardening results appear in \S\ref{sec:bounded_structural}, and \S\ref{sec:injective_escape} specifies the production escape encoding. These results concern \textsc{naive} candidate coverage, not protection against the \textsc{aware} adversary, which can generate transformed bytes directly.

\section{Experimental Methodology}
\label{sec:methodology}

Building on the DB-only adversary model in \S\ref{sec:adversary_model}, this section specifies the platforms, datasets, ARMOR tiers, and attack procedures used to evaluate ARMOR. We first describe the GPU cracking and IoT/edge resource-constrained benchmarking environments, then define the RockYou-derived targets, rule files, and transformation library. The section then details the \textsc{naive} dictionary/brute-force/rule-based attacks, the \textsc{aware} chain-expanded evaluation, and digit-only bypass experiments for high-bit and production \chain{ESC} hardening. The historical offline-resistance evaluation spans two GPUs (NVIDIA A100 and Blackwell), the \textsc{naive} and \textsc{aware} attack classes and the digit-class hardening study, RockYou-derived target populations, four hash algorithms, public rule sets, and fixed time budgets of $30$/$60$/$300$/$600$\,s (Table~\ref{tab:eval_scale_S}).

\subsection{High-Performance Cracking Platform}\label{sec:gpu}
Historical adversary throughput is measured on an NVIDIA A100 $80$\,GB PCIe GPU ($6{,}912$ CUDA cores), reflecting the offline-cracking threat model of \S\ref{sec:adversary_model}. Hardened-configuration replication uses an NVIDIA RTX~PRO~6000 Blackwell GPU ($24{,}064$ CUDA cores, $96$\,GB GDDR7). Section~\ref{sec:bounded_structural} reports the resulting digit-bypass and hardening measurements.

\subsection{IoT and Edge Evaluation Platforms}
\label{sec:iot_edge_platforms}
We benchmark ARMOR's defender-side overhead on a Raspberry Pi~3 Linux-class edge platform and three bare-metal IoT boards: the NUCLEO-L432KC (Cortex-M4), Raspberry Pi Pico~2~W (Cortex-M33), and ESP32 DevKit (Xtensa LX6). Table~\ref{tab:iot_edge_platforms} provides their hardware specifications and software environments. The Pi~3 serves as the reference where KDFs fit but are operationally costly. The three bare-metal boards represent classes where KDFs do not fit at security-meaningful settings (\S\ref{sec:kdf_infeasibility}).

\subsection{Software Environment and Attack Modes}
GPU cracking ran hashcat v$6.2.6$ on Ubuntu $24.04$ with optimized kernels (\lit{-O}). The main \textsc{naive} evaluation uses the following three attack modes against the same target files, with fixed wall-clock budgets of $30$/$60$/$300$/$600$\,seconds (\S\ref{sec:results_attack}). Table~\ref{tab:eval_scale_S} gives the per-block accounting.

\begin{enumerate}
\item \textbf{Dictionary.} Direct RockYou wordlist attack ($14.3$M passwords), measuring base candidate coverage.
\item \textbf{Brute-force.} Time-limited mask enumeration over $\{a\text{-}z, A\text{-}Z, 0\text{-}9\}$ plus hashcat's built-in special charset (\lit{?s}) for lengths $6$--$25$ using \lit{--increment}.
\item \textbf{Rule-based.} RockYou augmented with six primary hashcat-compatible rule files plus the extended \textit{omgwtfbbq} evaluation ($6.72$\,M rules), spanning five orders of magnitude in rule count (Table~\ref{tab:rulefiles}). This is the main \textsc{naive} attack mode because ARMOR targets rule-pipeline disruption. The digit-class bounded evaluation extends this set to $16$ rule files (\S\ref{sec:digit_class_methodology}, Table~\ref{tab:digit_perfile_S})~\cite{oclhashcat2013rule}.
\end{enumerate}

Secondary checks evaluate non-rule guessers on plain-chain and escape-chain (\lit{ESC-HC1}) digit-class targets. Markov-ordered masks and PRINCE use $600$\,s budgets. Hashcat masks (\lit{-a 3}) span lengths $1$--$8$, permitting digits (\lit{0--9}), lowercase letters (\lit{a--z}), or their union at each position, with default Markov ordering~\cite{narayanan2005fast}. PRINCE (PRobability INfinite Chained Elements) uses \lit{princeprocessor} over the top $100$K RockYou entries~\cite{steube2014prince}. PCFG (probabilistic context-free grammar) uses RockYou-trained \lit{pcfg\_cracker} with an untransformed-target control~\cite{weir2009password}. Its retained record describes an approximately $50$M-candidate file for hashcat checking. Elapsed time and stopping budget were not retained. These variants yield $11$ attack/target-group outcomes. All emit printable candidates and cannot reach non-printable escape output (\S\ref{sec:bounded_structural}).

\subsection{Dataset Construction}
Target sets were drawn from RockYou passwords filtered to lengths $6$--$25$ as
reproducible seed-$42$ random samples without replacement. We evaluate three
sizes, $|\mathcal{D}| \in \{10{,}000,\ 100{,}000,\ 1{,}000{,}000\}$, and reuse
each sampled set across hash algorithms and ARMOR configurations so the
target selection is held fixed across the hash-algorithm and chain comparisons. Every target plaintext appears in the full source corpus, giving the direct full-corpus wordlist attack complete base-word coverage. This does not mean every time-limited mask or rule attack generates every target. All ARMOR-resistance results in \S\ref{sec:results} use these target files.

\subsection{ARMOR Tier Structure}

The broad \textsc{naive} rule-based evaluation spans $25$ base configurations across five ARMOR tiers (T0--T4) and an untransformed baseline. These configurations draw on $12$ primitives, consisting of structural string transforms and rotate-only bit-diffusion operations (\chain{BSL}/\chain{BSR}). Adding the high-bit variants \chain{BSL-X}/\chain{BSR-X} (\S\ref{sec:bounded_structural}) gives the $|\mathcal{O}|{=}14$ library that the \textsc{aware} adversary enumerates (\S\ref{sec:aware_methodology}) and brings the full set to $31$ configurations (Table~\ref{tab:tier_configs}). The six additions are the two T0 variants and the hardened T3-HC1/2/3 and T4-HC1 chains, whose \chain{BSL-X}/\chain{BSR-X} steps mitigate the digit-only bypass. Production \chain{ESC} (\S\ref{sec:injective_escape}) closes this bypass. The representative ESC-HC1 chain is evaluated in the digit-class study (\S\ref{sec:digit_class_methodology}) and benchmarked for latency (\S\ref{sec:results_perf}). Each relevant configuration is evaluated against \lit{MD5}, \lit{SHA-1}, \lit{SHA-256}, and \lit{SHA-512}, with \textit{Baseline} (\lit{Hash(password)}) as the reference.

\subsection{Performance Benchmarking}
We measure transform latency on a Linux-class edge reference and bare-metal IoT boards (\S\ref{sec:iot_edge_platforms}). We also probe \lit{Argon2id} memory feasibility. Section~\ref{sec:secure_service_methodology} separately measures secure-service components.

\textbf{Linux reference (Pi~$3$).} We measure ARMOR pre-processing latency for lengths $\{6,11,14,20\}$, subtracting a \lit{memcpy}-only baseline to isolate transform cost. The password hash is common to all configurations. The same procedure times legacy hashes and \lit{Argon2id} profiles at embedded-tuned, OWASP-recommended $(19\,\mathrm{MiB},2,1)$~\cite{owaspPasswordStorageCheatSheet}, and high-memory settings. Table~\ref{tab:armor_overhead_S} reports the protocol and baseline-subtracted ARMOR costs.

\textbf{Bare-metal IoT benchmarks.} On the three bare-metal IoT boards in \S\ref{sec:iot_edge_platforms}, ARMOR latency is measured with platform cycle counters across the same four length bands. Warmup, iteration, and interrupt-masking details are reported in the note to Table~\ref{tab:embedded_latency}. The worst-case $4$-op chain, T4-HC1 at the longest length, defines peak embedded latency.

\textbf{Empirical \lit{Argon2id} infeasibility.} Each bare-metal IoT board ran a heap-allocation ceiling probe against the evaluated OWASP-recommended $(19\,\mathrm{MiB},2,1)$ \lit{Argon2id} profile (\S\ref{sec:kdf_infeasibility}). The largest successful allocations are reported in \S\ref{sec:results_perf}.

\begin{table*}[!t]
\renewcommand{\arraystretch}{1.3}
\setlength{\tabcolsep}{4pt}
\caption{ARMOR configuration library: $31$ configurations over $|\mathcal{O}|{=}14$ primitives.}
\label{tab:tier_configs}
\centering
\footnotesize
\begin{tabularx}{\textwidth}{@{}l l c >{\raggedright\arraybackslash}X@{}}
\toprule
\textbf{Tier} & \textbf{Description} & \textbf{$n$} & \textbf{Configurations / Chains evaluated} \\
\midrule
Baseline & Legacy hash & $1$  & \chain{Hash(password)} (reference with no ARMOR) \\
T$0$     & Single primitive      & $14$ & \chain{L-CASE, U-CASE, CAP, CAP-INV, TOG, REV, RFL, DUP, ROT-L/R, BSL/R, BSL/R-X} \\
T$1$     & $2$-op structural     & $3$  & \textbf{T1-C1}: \chain{TOG}$\to$\chain{REV}; \textbf{T1-C2}: \chain{CAP}$\to$\chain{DUP}; \textbf{T1-C3}: \chain{ROT-L}$\to$\chain{RFL} \\
T$2$     & $2$-op bit-diffusion  & $3$  & \textbf{T2-BC1}: \chain{BSL}$\to$\chain{ROT-L}; \textbf{T2-BC2}: \chain{BSR}$\to$\chain{REV}; \textbf{T2-BC3}: \chain{BSL}$\to$\chain{BSR} \\
T$3$     & $3$-op mixed          & $6$  & \textbf{T3-BC1}: \chain{ROT-R}$\to$\chain{REV}$\to$\chain{BSL}; \textbf{T3-BC2}: \chain{DUP}$\to$\chain{BSR}$\to$\chain{ROT-L}; \textbf{T3-BC3}: \chain{TOG}$\to$\chain{RFL}$\to$\chain{BSL}\newline \textbf{T3-HC1}: \chain{ROT-R}$\to$\chain{REV}$\to$\chain{BSL-X}; \textbf{T3-HC2}: \chain{DUP}$\to$\chain{BSR-X}$\to$\chain{ROT-L}; \textbf{T3-HC3}: \chain{TOG}$\to$\chain{RFL}$\to$\chain{BSL-X} \\
T$4$     & $4$-op mixed          & $4$  & \textbf{T4-BC1}: \chain{ROT-L}$\to$\chain{BSR}$\to$\chain{RFL}$\to$\chain{TOG}; \textbf{T4-BC2}: \chain{DUP}$\to$\chain{BSL}$\to$\chain{REV}$\to$\chain{ROT-L}; \textbf{T4-BC3}: \chain{BSL}$\to$\chain{REV}$\to$\chain{BSR}$\to$\chain{DUP}$^{\dagger}$; \textbf{T4-HC1}: \chain{ROT-L}$\to$\chain{BSR-X}$\to$\chain{RFL}$\to$\chain{TOG} \\
\hline\hline
Prod$^{\ddagger}$ & Production escape & $-$ & \textbf{ESC-HC1}: \chain{TOG}$\to$\chain{DUP}$\to$\chain{ESC} \\
\bottomrule
\addlinespace[2pt]
\multicolumn{4}{@{}>{\raggedright\arraybackslash}p{\textwidth}@{}}{\scriptsize \textbf{Note:} \textbf{C}=structural; \textbf{BC}/\textbf{HC}=base/hardened bit-diffusion. $^{\dagger}$validator-accepted chain with non-adjacent algebraic cancellation. $^{\ddagger}$production \chain{ESC} chain over $\mathcal{O}_{\mathrm{inj}}$, not among the $31$.} \\
\hline
\end{tabularx}
\end{table*}

\subsection{ARMOR as a TrustZone-M Secure Service}
\label{sec:secure_service_methodology}

The bare-metal benchmarks (\S\ref{sec:iot_edge_platforms}) isolate ARMOR's raw
transform cost, but they do not exercise the isolation boundary that the security
argument depends on. The CR must be unreadable from the
side that stores password hashes (DR-4, \S\ref{sec:security_privilege}). To measure
secure-side chain lookup, transform, and peppering components, we implement a prototype service
on the Cortex-M33 (Pico~2W) under TF-M at
Isolation Level~$2$, with the profile-medium configuration and dummy provisioning.
The same silicon and clock as Table~\ref{tab:iot_edge_platforms} are used. Only the
software stack changes, from bare-metal firmware to a TF-M secure image.

\textbf{Partition.} ARMOR runs as a single Platform Security Architecture (PSA) Root-of-Trust (RoT) partition
exposing the secure function \lit{ARMOR\_HARDEN}. Isolation Level~$2$ protects secure execution from non-secure software and additionally protects the PSA-RoT domain from the Application-RoT domain.\footnote{TF-M isolation levels: \url{https://trustedfirmware-m.readthedocs.io/en/tf-mv2.2.0/design_docs/ff_isolation.html}.} ARMOR is placed in the PSA-RoT rather
than the Application-RoT domain for a measurement reason. At Isolation Level~$2$ this platform
grants privileged execution only to PSA-RoT partitions. The in-enclave cycle counter used
below lives in the privileged-only private peripheral bus. The service compiles the
\emph{same} canonical \lit{armor\_core} source used in every other measurement in this
paper, unmodified, into the Secure Processing Environment (SPE). The non-secure
application reaches it only through the Secure Partition Manager by a single
\lit{psa\_call}. The per-account chain remains secure-side. The timing service returns the transformed output to its non-secure caller, so this prototype does not establish that a deployed verifier keeps that output inside its protected boundary. The per-account chain is fetched inside the SPE from Internal
Trusted Storage (ITS) under dummy provisioning, keyed by \lit{account\_id}, which is the concrete realization of
the CR isolation the analysis assumes for chain storage. The chain material is a secure-side asset,
addressable only through the service interface.

\textbf{Peppering comparator on the same platform.} To compare ARMOR against per-user
peppering (\S\ref{sec:peppering}) through the same secure-service implementation, the same
partition exposes a second operation that computes HMAC with \lit{SHA-256} over the
password through the platform PSA Crypto API using a $128$-bit key. The global-pepper key is imported once during
one-time setup and held live. The per-user pepper key is read from secure storage on
each call. Its compute timing includes key import, HMAC computation, and key destruction.
ARMOR validates the retrieved chain on each call, but its compute timing includes only the transform.
Chain validation is excluded from the component results in Table~\ref{tab:secure_service} (\S\ref{sec:results_secure_service}).
Both implementations use the same hardware and non-secure-to-secure interface.
The table therefore compares selected implementation components, not complete service execution times.
Routing peppering through PSA
Crypto, rather than the portable software \lit{SHA-256} of \S\ref{sec:results_perf},
includes PSA key lookup and context setup. We verified that this build does not reach the on-die \lit{SHA-256} accelerator. The RP2350 hardware driver is linked only for the attestation hardware abstraction layer (HAL), and no
\lit{SHA-256} alternative implementation is registered with mbedTLS, so the comparator
executes in software.

\textbf{Measurement.} Each operation returns an in-enclave cycle count taken from the
Cortex-M33 Data Watchpoint and Trace (DWT) cycle counter around its compute component, with the timing scope specified above.
Boundary overhead is measured separately by a no-op service
call, timed inside the enclave as the interval between consecutive service entries under
back-to-back invocation. The non-secure DWT is not used because its counter is not
reliable across the Security Attribution Unit and Implementation Defined Attribution Unit boundary. That interval estimates the return-and-reentry component under back-to-back no-op calls. Following the bare-metal protocol (note to
Table~\ref{tab:embedded_latency}), we report the minimum of $N{=}256$ calls after a
warmup, \rev{using $k{=}1$ and validator-accepted, maximum-expansion
$k{=}5$ chains as measured reference configurations}. All measurements use a $15$-character password. Each reported sum combines one no-op boundary estimate, the applicable secure-storage read, and keyed compute, measured separately. These sums exclude ARMOR chain validation, credential DB lookup, final password hashing, and verifier comparison. They are not complete authentication times or latency distributions. The service is exercised from a purpose-built non-secure timing application.
No measurement path depends on secure-side logging.

\subsection{ARMOR-Aware Adversary Methodology}
\label{sec:aware_methodology}

\rev{We complement the \textsc{naive} evaluation with an \textsc{aware} enumerator (\S\ref{sec:adversary_model}). Transformed-target coverage is the fraction of targets matched by a base word under any hypothesized chain. A match need not recover the original password or authenticate under the assigned chain. The enumerator applies all $14^3 = 2{,}744$ depth-$3$ orderings to each candidate and compares transformed bytes directly. There is no time budget or additional candidate-length limit beyond the eligible-corpus filter. Enumeration is exhaustive within the sampled population and each dictionary prefix.}

\rev{The target population is fixed at $1{,}000{,}000$ accounts sampled once from the corpus with seed $42$. It and the three hardened T$3$ chains of Table~\ref{tab:tier_configs} remain fixed across every point. Corpus bytes are used as they appear in the canonical file, with no re-encoding. The independent variable is the attacker's base dictionary, taken as a prefix of the eligible corpus at $10^5$, $10^6$, $5{\times}10^6$ and the full corpus. We distinguish requested prefix entries, eligible entries after the length filter, and unique candidates after deduplication. The full corpus supplies $14{,}053{,}727$ eligible entries and $14{,}053{,}283$ unique candidates. Coverage is weighted by sampled account, preserving target multiplicity. Distinct-transformed-digest coverage is given only as a bridge to the earlier unsalted experiment. Hypotheses evaluated are reported separately from dictionary size.} Deeper tiers add one $|\mathcal{O}|$ factor per operation.

\textbf{Conditional guessing number.} Enumerating the chain space measures how many candidates exist,
not how many tests are needed within the hidden-chain subproblem. \rev{For the production library
$\mathcal{O}_{\mathrm{inj}}$ we enumerate every length-$k$ sequence over the eight non-terminal
production operations at $k{=}1$--$5$. For each fixed plaintext $p$, we group accepted chains by
identical transformed output, assign each output class its hidden-chain multiplicity, order the
classes by that mass, and compute the expected rank of the class containing the assigned chain. This
computes $G_k(p)$ of Definition~V.1. We then account-weight its distribution
over the corpus and report the median and $5$th percentile. It does not order unknown base plaintexts
or quantify the computational cost of salted verifier checks.} The corpus is
the full RockYou list filtered to lengths $6$--$25$ ($14{,}053{,}727$ eligible, $14{,}053{,}283$
unique). The measurement is deterministic (no sampling, no seed) and the corpus, operation set,
and per-depth chain sets are pinned by hash.

\textbf{Bounded production attack and credential recovery.}
We additionally evaluate the first $N{=}1{,}000$ retained targets from a $10{,}000$-target set whose passwords occur in the first $10^6$ eligible RockYou entries ($999{,}997$ unique passwords). Targets follow a deterministic SHA-256 counter/rejection sample of eligible entries without replacement, using seed \lit{ARMOR-W2-production-targets-v1}. The bounded prior is nearly flat: $999{,}994$ candidates occur once; \lit{daniel}, \lit{te amo}, and \lit{i love u} each occur twice. The experiment uses assigned chains from the $422$ accepted $k{=}3$ chains with fixed terminal \chain{ESC}; the ordering model assumes independent uniform chain assignment. Both arms use the same passwords and $16$-byte salts; their verifiers are $\mathrm{MD5}(\mathrm{ESC}(T_K(P))\Vert s)$ and $\mathrm{MD5}(P\Vert s)$. A CPU enumerator merges identical transformed values and orders them by descending aggregate candidate frequency times chain multiplicity, with lexicographic byte ties. Plain candidates are ordered by descending frequency, then first occurrence. This experiment uses A6000 and A100 GPUs, separately from the historical A100/Blackwell evaluation. CUDA replays this fixed unit-cost order using reference and optimized MD5 paths, with length-appropriate short and general kernels. Batch-start snapshots of unrecovered accounts count all executed salted checks, including within-batch overshoot. A minimum of $200$ batches, one warm-up and three measured repeats are used on each GPU. Per-account ranks and work reconcile exactly with the frozen CPU-verified reference. A6000 and A100 use the same hash-bound inputs and deterministic exports.

Kernel events measure MD5 search work; search-wall timing measures the recovery loop. Both exclude CPU schedule construction and sorting. Transfer counters overlap some search operations and are not blindly added to search wall. A separate optimized Hashcat mode-$10$ plain-only anchor checks recovered bytes; different work and process timing prevent a matched-throughput comparison. Full ambiguity sets are then enumerated over the bounded pool and ranked by candidate frequency times consistent-chain multiplicity. A separate grader simulates password submissions under each account's assigned chain. These retained October~10 same-target Threadripper 9960X CPU results are reused, not rerun as part of either new GPU suite (Appendix~\ref{app:VIII}).

\subsection{Digit-Class Bypass Characterization}
\label{sec:digit_class_methodology}

The bounded structural-disruption claim in \S\ref{sec:structural_invalidation}
identifies a plaintext class for which Layer-$2$ bit rotation can remain
printable. We characterize this class with three procedures.

\textbf{1. Structural analysis.} We compute byte-level mapping under \chain{BSL} for all \lit{$0x00$--$0xFF$} inputs. ASCII digits \lit{$0x30$--$0x39$} map to the printable bytes \lit{$0x60$}, \lit{$0x62$}, \ldots, \lit{$0x72$}, so digit-only plaintexts can remain reachable by printable rule pipelines under T$2$ bit-rotation chains. The hardened \chain{BSL-X}/\chain{BSR-X} primitives (rotate then \chain{OR}~\lit{0x80}, \S\ref{sec:bounded_structural}) force every byte, including digits, outside printable ASCII. Injective \chain{ESC} closes the bypass fully (\S\ref{sec:injective_escape}). Boot-time self-tests verify, for all $256$ inputs on each board, that \chain{BSL-X}/\chain{BSR-X} leave the printable-ASCII range and \chain{ESC} emits non-printable byte pairs.

\textbf{2. Empirical prevalence.} We measure the fraction of digit-only passwords in the length-filtered RockYou corpus across six sampling windows: top $100$, $500$, $2{,}000$, $10{,}000$, $100{,}000$, and $1{,}000{,}000$ entries. The digit-only fraction ranges from $4\%$ to $19\%$ across these windows (Table~\ref{tab:digit_prevalence}). These values describe the evaluated digit-only subset, not the wider printable-output class (\S\ref{sec:bounded_structural}).

\textbf{3. Per-rule-file digit-class evaluation.} We evaluate $100{,}000$ digit-only T$3$ targets (seed-$42$ RockYou sample) against $16$ public rule files, extending the seven of Table~\ref{tab:rulefiles} to the full digit-class spectrum. Each attack runs to exhaustion over wordlists capped at $10^{6}$ entries on the Blackwell GPU (\S\ref{sec:gpu}), comparing the non-hardened baseline against the production \chain{ESC} form on both digit and non-digit targets (\S\ref{sec:bounded_structural}). Because \chain{ESC} is injective and non-printable, recovery does not depend on printable dictionary content, as the full per-rule-file results confirm (Table~\ref{tab:digit_perfile_S}).

\section{Results and Analysis}
\label{sec:results}
This section reports ARMOR's empirical resistance and IoT/edge overhead. We begin with baseline cracking (\S\ref{sec:results_baseline}), then evaluate \textsc{naive} resistance across attack types, rule sets, chain depths, and sensitivity factors (\S\ref{sec:results_attack}--\S\ref{sec:results_sensitivity}). We next quantify \textsc{aware} transformed-target coverage (\S\ref{sec:results_aware}), evaluate escape encoding against the digit-only bypass (\S\ref{sec:bounded_structural}), measure IoT-class hardware overhead (\S\ref{sec:results_perf}), and compare secure-service components (\S\ref{sec:results_secure_service}). We conclude by assessing the design requirements (\S\ref{sec:results_summary}; Table~\ref{tab:design_requirements}).

\subsection{Baseline Vulnerability}
\label{sec:results_baseline}

Under the evaluated offline-cracking threat, legacy hashes are highly exposed. Across valid rule-based cracking measurements, the average baseline crack rate is $86.52\%$ on the A100 GPU (Table~\ref{tab:summary}; Fig.~\ref{fig:tier_avg}), consistent with the high-throughput adversary model of \S\ref{sec:adversary_model}. With a $60$\,s budget, rule-based attacks recover the legacy \lit{MD5} and \lit{SHA-1} baseline hashes at $100\%$ even with \lit{B64}, the smallest evaluated rule file ($64$ rules), on the $10$K dataset.

\subsection{Attack Type and Rule-Set Penetration}
\label{sec:results_attack}
\begin{figure*}[!t]
  \centering
  \subfigure[Average crack rate by tier.]{%
    \includegraphics[width=0.32\linewidth]{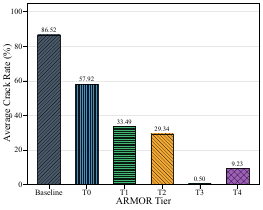}%
    \label{fig:tier_avg}}\hfill
  \subfigure[Crack rate by attack type.]{%
    \includegraphics[width=0.32\linewidth]{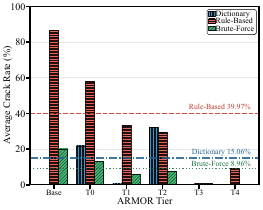}%
    \label{fig:tier_by_attack}}\hfill
  \subfigure[Crack-rate heatmap (tier $\times$ rule file).]{%
    \includegraphics[width=0.35\linewidth]{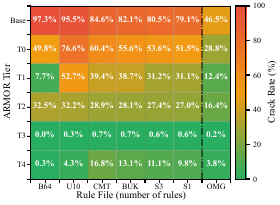}%
    \label{fig:rule_heatmap}}
  \caption{Baseline and \textsc{naive}-adversary attack surface across ARMOR tiers (A100).}
  \label{fig:rule_panel}
    \vspace{-5mm}
\end{figure*}

Rule-based attacks have the highest mean recovery rate across the evaluated configurations ($39.97\%$), followed by dictionary ($15.06\%$) and brute-force ($8.96\%$) attacks, consistent with modern offline-cracking practice. Crack rates are therefore reported primarily against rule-based attacks, with the per-attack breakdown in Fig.~\ref{fig:tier_by_attack}. All three attack types converge to near-zero at T3 against the \textsc{naive} adversary, indicating that structural disruption holds across attack modes when the attacker does not perform ARMOR-aware chain expansion.

Fig.~\ref{fig:rule_heatmap} shows that rule-set size alone does not explain recovery. Across a $64$--$1.27$M-rule span (\texttt{B64} to \texttt{S1}; Table~\ref{tab:rulefiles}), the legacy baseline changes by only $18$ percentage points (pp), while pre-hardened T3 stays near zero ($\leq 0.75\%$; Fig.~\ref{fig:S_naive_diagnostics}(a)). The separate digit-class evaluation identifies a printable-output bypass, suppressed by high-bit hardening and closed by \chain{ESC} under the evaluated stock tools (\S\ref{sec:bounded_structural}). Purely structural T1 chains remain partially vulnerable, while the collapse of pre-hardened T3 to near zero reflects structural disruption of the rule grammar for non-digit inputs rather than a timing-only slowdown.

\textbf{Extended rule-set evaluation.} The largest public rule file evaluated, \texttt{OMG} ($6.72$\,M rules; Table~\ref{tab:rulefiles}), follows the same pattern in the extended evaluation (Table~\ref{tab:eval_scale_S}): T3 recovery remains $0.21\%$ (mean across conditions). Bit-diffusion primitives stay near-zero (\chain{BSL} $0.65\%$, \chain{BSR} $0.03\%$), since no evaluated public rule set models circular bit-rotation. Large rule files, including \texttt{OMG}, contain native \lit{d}/\lit{f} rules, equivalent to ARMOR's \chain{DUP}/\chain{RFL}, so standalone Layer-1 deployments remain vulnerable and are not recommended.

\subsection{Single-Operation ARMOR (T0)}
\label{sec:results_t0}

The T0 tier reduces the \textsc{naive} crack rate from $86.52\%$ to $57.92\%$ on the GPU, with high within-tier heterogeneity (Fig.~\ref{fig:t0_rates}). Bit-diffusion primitives \chain{BSL}/\chain{BSR} resist best because the evaluated attacks do not model circular bit-rotation. \chain{DUP}/\chain{RFL} ($48.32\%$/$39.03\%$) suggest rule-pipeline disruption. The hardened \chain{BSL-X}/\chain{BSR-X} variants suppress the digit-class leakage that leaves base \chain{BSL} at $\approx 1.48\%$ (\S\ref{sec:bounded_structural}). \chain{L-CASE} ($86.69\%$) and other case-normalizing transforms remain vulnerable as expected. Layer-1 primitives vary character order, case, or length within composed chains. They are not intended as standalone defenses against the evaluated attacks.

\subsection{Chained ARMOR Resistance (T1--T4)}
\label{sec:t4_anomaly}

\begin{figure*}[!t]
  \centering
  \subfigure[Per-tier crack-rate CDF.]{%
    \includegraphics[width=0.33\linewidth]{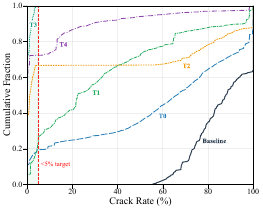}%
    \label{fig:cdf_dist}}\hfill
  \subfigure[T0 primitive rates.]{%
    \includegraphics[width=0.33\linewidth]{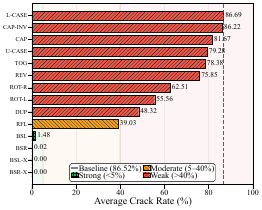}%
    \label{fig:t0_rates}}\hfill
  \subfigure[T4 chains.]{%
    \includegraphics[width=0.33\linewidth]{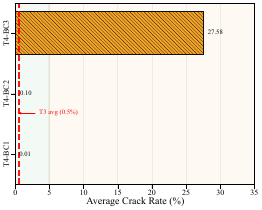}%
    \label{fig:t4_chains}}
  \caption{Chain-depth effects on rule-based \textsc{naive} recovery. Per-tier distributions shift toward $0\%$ as depth increases, with all evaluated rule-based T3 configurations below $5\%$. T4 results include an algebraically degenerate chain accepted by the adjacent-pair validator.}
  \label{fig:chained_resistance}
  \vspace{-5mm}
\end{figure*}

The central \textsc{naive}-resistance finding is a sharp crack-rate collapse with chain depth (Fig.~\ref{fig:tier_avg}). Per-tier rule-based distributions shift toward $0\%$ as depth increases (Fig.~\ref{fig:cdf_dist}). By T3, they form a near-zero spike (Fig.~\ref{fig:S_violin}). A T4 counterexample shows that the production validator's adjacent-pair check does not remove every degenerate chain (\S\ref{sec:chain_constraints}). The following rates are unweighted means over the evaluated rule-based configurations, not bounds on individual tests.

\begin{itemize}
\item \textbf{T1} (2-op structural): $33.49\%$. T1 improves over T0 but remains vulnerable to combined rule application.
\item \textbf{T2} (2-op with \chain{BSL}/\chain{BSR}): $29.34\%$ mean, inflated by inverse-pair \lit{T2-BC3} ($86.56\%$; Table~\ref{tab:perchain_inverse}), where \chain{BSR} cancels \chain{BSL}'s diffusion. The accepted \lit{T2-BC1}/\lit{T2-BC2} chains each have mean recovery $\leq 1.5\%$.
\item \textbf{T3} (3-op, pre-hardened \chain{BSL}/\chain{BSR}): $0.50\%$. The separate input-class experiment characterizes a digit-only bypass and its hardening (\S\ref{sec:bounded_structural}). The mixed-target results alone do not identify the plaintext class of every residual match.
\item \textbf{T4} (4-op): Two chains have near-zero mean recovery, while \lit{T4-BC3} has a mean of $27.58\%$ (Fig.~\ref{fig:t4_chains}). The validator accepts this chain because the inverse rotations are not adjacent. \chain{BSL} and \chain{BSR} rotate bits within each byte, whereas \chain{REV} changes only byte order. Reversal leaves their inverse relation unchanged, so \chain{BSL}$\to$\chain{REV}$\to$\chain{BSR}$\to$\chain{DUP} is equivalent to \chain{REV}$\to$\chain{DUP}.
\end{itemize}

The other two T4 chains each have mean rule-based \textsc{naive} recovery $\leq 0.10\%$. The accepted \lit{T4-BC3} counterexample shows the limit of adjacent-pair pruning, not a failure of fixed-chain injectivity.

\subsection{Sensitivity and Recovery-Rate Distribution}
\label{sec:results_sensitivity}

Secondary variables do not materially change the \textsc{naive}-resistance trend. Over $1{,}800$ rule-based cells per hash algorithm, the unweighted mean is $44.34\%$ for \lit{MD5} and $34.56\%$ for \lit{SHA-512} (Fig.~\ref{fig:S_naive_diagnostics}(b)). Across the $9{,}552$-cell grid of all attack types and configurations, the unweighted mean decreases $2.11$\,pp from $10$K to $1$M targets and increases $14.39$\,pp from $30$ to $600$\,s. Separately, rule-based T3 remains nearly flat (Fig.~\ref{fig:S_sensitivity_diagnostics}(a)--(b)). Distributional plots show all evaluated rule-based T3 configurations below $5\%$ recovery. The T4 outlier is validator-accepted \lit{T4-BC3}, whose non-adjacent rotations cancel (Fig.~\ref{fig:cdf_dist}; Fig.~\ref{fig:S_violin}).

\subsection{ARMOR-Aware Adversary Result}
\label{sec:results_aware}
\begin{table}[!t]
\centering
\caption{\rev{Account-weighted plain dictionary coverage and \textsc{aware} transformed-target coverage for $1{,}000{,}000$ fixed target accounts. A matching candidate-chain pair need not be a valid login credential. The hypotheses count includes all $14^3$ chains per base word.}}
\label{tab:aware_results}
\renewcommand{\arraystretch}{1.2}
\footnotesize
\setlength{\tabcolsep}{2.8pt}
\begin{tabular}{@{}lrrrr@{}}
\toprule
\textbf{Base dictionary $D$} & $10^5$ & $10^6$ & $5{\times}10^6$ & full \\
\midrule
Unique candidates      & $100{,}000$ & $999{,}997$ & $4{,}999{,}899$ & $14{,}053{,}283$ \\
Hypotheses evaluated   & $2.74{\times}10^{8}$ & $2.74{\times}10^{9}$ & $1.37{\times}10^{10}$ & $3.86{\times}10^{10}$ \\
\midrule
Plain dictionary coverage & $0.73\%$ & $7.11\%$ & $35.59\%$ & $100.00\%$ \\
\midrule
\textsc{aware} \lit{T3-HC1} & $1.46\%$ & $9.21\%$ & $38.22\%$ & $100.00\%$ \\
\textsc{aware} \lit{T3-HC2} & $2.94\%$ & $14.62\%$ & $45.70\%$ & $100.00\%$ \\
\textsc{aware} \lit{T3-HC3} & $1.66\%$ & $9.53\%$ & $38.47\%$ & $100.00\%$ \\
\bottomrule
\addlinespace[2pt]
\multicolumn{5}{@{}>{\raggedright\arraybackslash}p{0.95\columnwidth}@{}}{\scriptsize \textit{Note:} Exhaustive byte-domain enumeration over the canonical corpus ($14{,}053{,}727$ eligible entries, $14{,}053{,}283$ unique candidates) without a time budget or an additional candidate-length limit beyond the eligible-corpus filter. Dictionaries are prefixes of the eligible corpus. Distinct-digest coverage, the bridge to the earlier unsalted hashcat experiment, is $9.18\%$, $12.54\%$ and $9.52\%$ at $D{=}10^6$ for HC1/HC2/HC3.}\\
\end{tabular}
\vspace{-2mm}
\end{table}
Following \S\ref{sec:aware_methodology}, the \textsc{aware} evaluation applies all $14^3=2{,}744$ depth-$3$ chains to each base word in four eligible-corpus dictionary prefixes of $100$K, $1$M, $5$M, and the full corpus. It compares transformed bytes against one fixed population of $1$M sampled target accounts. At the $1$M prefix, this gives approximately $2.74$ billion chain hypotheses before cross-account amortization. Enumeration exhausts each prefix, so transformed-target coverage reflects dictionary content and chain equivalence rather than a time budget.

\textbf{Evaluation and production chain spaces.} The historical coverage evaluation enumerates all $14^3{=}2{,}744$ ordered chains in its evaluation library. Distinct chains can produce the same transformed candidate (chain equivalence, \S\ref{sec:chain_constraints}), so the nominal count is not a measured work factor. Production instead draws from the eight injective non-terminal operations, applies the production validator, and appends terminal \chain{ESC}. Its conditional output-class cost is reported separately as $G_k$ (Table~\ref{tab:guessing_number}, Table~\ref{tab:e3_distribution}). These are different libraries and quantities. Neither quantifies the computational cost of the joint unknown-plaintext, salt, and verifier-check attack.

\textbf{Conditional chain uncertainty: the guessing number.} The median conditional guessing number is approximately $52$ expected output-class tests at the primary evaluated depth $k{=}3$ and $317$ at $k{=}5$ (Table~\ref{tab:guessing_number}). Definition~\ref{def:conditional_guessing} in \S\ref{sec:operational_guarantees} specifies the conditional unit-cost model.
The distribution covers the full length-filtered RockYou corpus ($14{,}053{,}727$ eligible entries, $14{,}053{,}283$ unique strings) for the accepted production set $\mathcal{V}_k$. Table~\ref{tab:guessing_number} also reports the accepted-chain count $|\mathcal{V}_k|$ and its ratio to the median $G_k$. The observed corpus minima are $4.54$ expected class tests at $k{=}3$ and $6.56$ at $k{=}5$, derived from the retained class-multiplicity histograms. These observations are not universal worst-case guarantees. The general lower bound is $G_k(p)\geq1$ (\S\ref{sec:operational_guarantees}).
The bounded production experiment below measures salted work and CUDA search time; the minimum cost of an optimized end-to-end attack remains unestablished. Within its conditional unit-cost model, the analysis quantifies chain uncertainty after accounting for equivalent outputs, given the correct base plaintext.

\begin{table}[!t]
\renewcommand{\arraystretch}{1.15}
\caption{Conditional unit-cost guessing number $G_k$ and validator-accepted production chain count over the eligible RockYou entries.}
\label{tab:guessing_number}
\centering
\footnotesize
\setlength{\tabcolsep}{4pt}
\begin{tabular}{@{}rrrrr@{}}
\toprule
$k$ & $|\mathcal{V}_k|$ & $\mathrm{med}_{\mathrm{acct}}[\log_2 G_k]$ & $\mathrm{med}_{\mathrm{acct}}[G_k]$ & $|\mathcal{V}_k|/\mathrm{med}_{\mathrm{acct}}[G_k]$ \\
\midrule
$1$ & $8$        & $2.17$ & $4.5\times$   & $1.8\times$  \\
$2$ & $58$       & $4.01$ & $16.1\times$  & $3.6\times$  \\
$3$ & $422$      & $5.71$ & $52.2\times$  & $8.1\times$  \\
$4$ & $3{,}054$  & $7.10$ & $137.0\times$ & $22.3\times$ \\
$5$ & $21{,}822$ & $8.31$ & $316.6\times$ & $68.9\times$ \\
\bottomrule
\addlinespace[2pt]
\multicolumn{5}{@{}>{\raggedright\arraybackslash}p{0.95\columnwidth}@{}}{\scriptsize \textit{Note:} $\mathrm{med}_{\mathrm{acct}}[\cdot]$ is weighted by multiplicity in the $14{,}053{,}727$ eligible corpus entries over $14{,}053{,}283$ unique candidate strings; it does not use account counts from the original breach. $G_k(p)$ is the conditional guessing number of Definition~\ref{def:conditional_guessing} (expected unit-cost output-class tests, given the correct base plaintext). Column~3 is $\mathrm{med}_{\mathrm{acct}}[\log_2 G_k(p)]$ in bits; column~4 is $2^{\mathrm{med}_{\mathrm{acct}}[\log_2 G_k(p)]}$ in expected class tests at full precision, which the rounded column~3 reproduces only approximately. The $5$th percentile of $\log_2 G_k(p)$ is $2.17$, $3.93$, $5.48$, $6.75$, and $7.84$ bits for $k{=}1$--$5$. Full RockYou corpus, lengths $6$--$25$.}\\
\end{tabular}
\vspace{-3mm}
\end{table}

\textbf{Expansion by input class.} Digit-only passwords are a special case because case maps act as the identity. This class also causes the digit bypass (\S\ref{sec:bounded_structural}). Escape encoding closes the printable bypass. Historical transformed-target coverage enumerates the full evaluation library and does not depend on production-validator pruning.

\begin{revblock}
\textbf{Transformed-target coverage grows with the dictionary.} Plain coverage in Table~\ref{tab:aware_results} rises $0.73\% \to 7.11\% \to 35.59\% \to 100\%$. Transformed-target coverage also increases and reaches $100\%$ for all three chains at the full eligible corpus. The full dictionary contains every sampled password, but the enumerator records any matching chain. It does not distinguish the original password from another candidate that matches only under a different chain.

\textbf{Equal dictionary size is not equal attacker effort.} The two series in Table~\ref{tab:aware_results} use the same dictionary sizes but not the same attacker effort. For each base word the \textsc{aware} attacker evaluates all $14^3$ chains, so at $D{=}10^6$ it performs $2.74\times10^{9}$ hypothesis evaluations against the plain attacker's $10^6$. The table is a transformed-target coverage curve, not an attacker-cost curve. \rev{Table~\ref{tab:guessing_number} separately measures expected unit-cost output-class tests given the correct base plaintext. It does not determine the order of plaintext guesses or quantify the computational cost of salted verifier checks.}

\textbf{Why transformed-target coverage exceeds plain coverage.} The gap is largest for \lit{T3-HC2} ($2.94\%$ against $0.73\%$ at $D{=}10^5$). A base word can match under a chain other than the one assigned. Non-injective historical primitives also permit within-chain collisions, notably the $2$-to-$1$ \chain{BSR-X}. Production excludes those primitives, but injectivity under a fixed chain does not exclude cross-chain matches. The gap therefore does not establish additional login-valid password recovery.
\end{revblock}

\textbf{Chain-selection implication.} \lit{T3-HC2} has the largest transformed-target coverage (\rev{$14.62\%$ account-weighted at $D{=}10^6$, Table~\ref{tab:aware_results}}). Its \chain{DUP} followed by non-bijective \chain{BSR-X} yields $\approx 3.99\%$ fewer distinct transformed values than distinct plaintexts, versus $\le 0.08\%$ for the near-injective \chain{BSL-X}-terminated chains (Appendix~\ref{app:III}). These within-chain collisions motivate the production injectivity requirement. They do not make arbitrary-chain matches proof of credential recovery. The \chain{BSL-X}/\chain{BSR-X} chains remain historical evaluation configurations, separate from the eight-operation production library with terminal \chain{ESC} (\S\ref{sec:injective_escape}).


\textbf{Measured bounded production attack.} Table~\ref{tab:bounded_gpu} reports the custom optimized CUDA replay on two GPUs. At full Stage-1 recovery, ARMOR executes $46{,}152{,}558{,}665$ salted evaluations versus $516{,}541{,}868$ for plain MD5, a ratio of $89.35$. Both GPUs, both kernels and all repeats produce identical account ranks and counted work. This measures one implemented production-chain attack, separately from conditional $G_k$ and the historical coverage curve. All sampled passwords are in the bounded dictionary, so $100\%$ recovery is by construction; the ratio is not a universal lower bound on attacker cost. The schedule is ordered under a unit-cost model, not optimized for GPU-dependent per-candidate cost.

\begin{table}[!t]
\centering
\caption{Bounded production attack on $1{,}000$ paired salted-MD5 targets: optimized CUDA search. Times in seconds are means $\pm$ sample SD over three measured repeats; ratios are ratios of means.}
\label{tab:bounded_gpu}
\footnotesize
\setlength{\tabcolsep}{3pt}
\begin{tabular}{@{}llrrr@{}}
\toprule
GPU & Boundary & ARMOR & Plain & Ratio\\
\midrule
A6000 & Kernel & $9.22875\!\pm\!0.04717$ & $0.07297\!\pm\!0.00026$ & $126.48$\\
A6000 & Search wall & $9.23296\!\pm\!0.04716$ & $0.07701\!\pm\!0.00035$ & $119.89$\\
A100 & Kernel & $7.35657\!\pm\!0.00033$ & $0.06102\!\pm\!0.00176$ & $120.56$\\
A100 & Search wall & $7.36911\!\pm\!0.00068$ & $0.06914\!\pm\!0.00256$ & $106.58$\\
\bottomrule
\end{tabular}
\end{table}

\textbf{From transformed matches to login-valid passwords.} Complete inverse enumeration gives $908$ singleton plaintext sets among the $1{,}000$ recovered targets. Posterior-ordered, simulated login trials recover every original and login-valid password with mean $1.086$, median $1$, and maximum $6$ submissions. Enumeration takes $1.196958$\,s on the previously measured CPU, including $422{,}000$ chain inversions, a corpus-frequency pass and forward checks. These reused same-target results are conditional on the bounded pool and uniform-chain prior. They do not measure live authentication latency or justify adding simulated submissions to offline seconds.

\textbf{Attack-cost scope.} CPU grouping and sorting cost $106.5$\,s on the A6000 host and $233.7$\,s on the A100 host, before subsequent export overhead. Table~\ref{tab:bounded_gpu} excludes that preparation; its reuse must be stated when amortizing across target sets. Reference/optimized reconciliation and a plain-only Hashcat anchor support correctness, but the anchor has different work and timing boundaries. The measurements do not establish fastest attainable attack throughput, a KDF-general multiplier, energy or monetary cost, or superiority to an isolated high-entropy pepper. Appendix~\ref{app:VIII} gives checkpoints, environment, and cost categories. Repeat SD describes timing variation on one target set, not uncertainty across independently sampled populations.

\subsection{Bounded Structural Claim and Digit-Class Bypass}
\label{sec:bounded_structural}
The separate digit-class evaluation identifies a printable-output bypass for digit-only plaintexts under per-byte bit rotation. This subsection characterizes that class, quantifies its exploitability, and shows that production \chain{ESC} hardening closes it under the evaluated \textsc{naive} attacks.

\textbf{Failure mode.} Digit-only inputs are the evaluated bypass class, not the only inputs that can remain printable. Digits \lit{0}--\lit{9} (bytes \lit{0x30}--\lit{0x39}) map under \chain{BSL} to printable bytes $\{\lit{0x60}, \lit{0x62}, \lit{0x64}, \ldots, \lit{0x72}\}$, namely the backtick and every other lowercase letter (\lit{b}, \lit{d}, \ldots, \lit{r}). More generally, bytes \lit{0x20}--\lit{0x3f} remain printable after \chain{BSL}. We measured digit-only recovery, not the prevalence or recovery of this wider class. Per-rule-file results are in Table~\ref{tab:digit_perfile_S}.

\begin{figure}[!t]
  \centering
  \includegraphics[width=\linewidth]{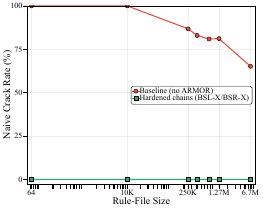}
  \caption{Hardened \textsc{naive} crack rate across rule-file sizes.}
  \label{fig:hardened_rulesize}
  \label{fig:hardening_panel}
  \vspace{-3mm}
\end{figure}

\begin{table}[!t]
\centering
\caption{Digit-class recovery by stock \textsc{naive} guessers.}
\label{tab:nonrule_guessers}
\renewcommand{\arraystretch}{1.15}
\footnotesize
\begin{tabular}{@{}lrr@{}}
\toprule
\textbf{Guesser} & \textbf{Unhardened} & \textbf{ESC-hardened} \\
\midrule
Rule-based hashcat & $45.3\%$ & $0\%$ \\
Markov-ordered masks & $30.9\%$ & $0\%$ \\
PCFG & $0.67\%$ & $0\%$ \\
PRINCE & $0.06\%$ & $0\%$ \\
\bottomrule
\addlinespace[2pt]
\multicolumn{3}{@{}>{\raggedright\arraybackslash}p{0.95\columnwidth}@{}}{\scriptsize \textit{Note:} Recovery on $100$K digit-class targets. The rule-based maximum covers $16$ rule files and wordlists up to $10^6$ entries, run to exhaustion (\S\ref{sec:digit_class_methodology}). Masks/PRINCE use $600$\,s budgets. PCFG records describe a $\sim 50$M-candidate file with unrecorded time/budget (\S\ref{sec:methodology}). These are not equal-effort comparisons. Zeros describe the evaluated stock tools, not \textsc{aware} chain search.}\\
\end{tabular}
\end{table}

\textbf{Bypass severity and hardening.} Without hardening, a \textsc{naive} attack using hashcat across $16$ public rule files (through \texttt{OMG}, $6.72$M rules) and dictionaries up to $10^6$ words recovers as much as $45.3\%$ of digit-only T3 targets. Per-byte rotation maps digits into printable ASCII, enabling rule substitutions. Recovery on the evaluated non-digit targets is $\leq 0.10\%$. Digit-only passwords account for $4$--$19\%$ of length-filtered RockYou (Table~\ref{tab:digit_prevalence}). Table~\ref{tab:hardening_tradeoff_main} compares recovery and injectivity across hardening forms. Appendix~\ref{app:II} explains the high-bit OR/XOR mechanisms.

\begin{table}[!t]
\centering
\caption{Maximum \textsc{naive} recovery by hardening form.}
\label{tab:hardening_tradeoff_main}
\footnotesize
\renewcommand{\arraystretch}{1.15}
\setlength{\tabcolsep}{3pt}
\begin{tabularx}{\columnwidth}{@{}>{\raggedright\arraybackslash}Xrrc@{}}
\toprule
\textbf{Hardening form} & \textbf{Digit} & \textbf{Non-digit} & \textbf{Injective}\\
\midrule
None (\chain{BSL}/\chain{BSR}) & $45.3\%$ & $0.10\%$ & N/A\\
High-bit OR ($\lor\,\lit{0x80}$) & $0.10\%$ & $0.04\%$ & no\\
\chain{BSL-XOR}/\chain{BSR-XOR} ($\oplus\,\lit{0x80}$) & $0.10\%$ & $1.08\%$ & yes\\
\chain{BSL-XOR}/\chain{BSR-XOR} $+$ length expansion & $0.003\%$ & $0.17\%$ & yes\\
\chain{ESC} (escape encoding) & $\mathbf{0.000\%}$ & $\mathbf{0.000\%}$ & yes\\
\bottomrule
\addlinespace[2pt]
\multicolumn{4}{@{}>{\raggedright\arraybackslash}p{\columnwidth}@{}}{\scriptsize \textit{Note:} Maxima for the evaluated digit and non-digit targets. Zero recovery concerns the tested \textsc{naive} attacks, not \textsc{aware} chain search. The hardening mechanisms are described in Appendix~\ref{app:II}.}\\
\end{tabularx}
\end{table}

\textbf{Escape-encoding result.} Evaluated attacks recover $0/100$K \chain{ESC}-hardened targets across all $384$ configurations ($3$ chains $\times$ $16$ rule files $\times$ $4$ wordlist sizes, $1$K--$1$M, $\times$ $2$ target classes). Table~\ref{tab:digit_perfile_S} reports per-rule-file outcomes. The per-configuration $95\%$ upper bound is $\leq 0.003\%$.

\textbf{Validation and mechanism.} The broader hardened evaluation covers $588$ configurations on both GPUs and is consistent with Table~\ref{tab:hardening_tradeoff_main}. Mixed-target \textsc{naive} recovery stays $\leq 0.01\%$ flat across rule-file sizes (Fig.~\ref{fig:hardened_rulesize}) and attack budgets (Fig.~\ref{fig:S_sensitivity_diagnostics}(c)), while the unprotected baseline spans $22.5$--$100\%$. Per-crack rule logging attributes residual digit-class recovery to high-bit substitution rules on repeated-character words. On the $10$-symbol digit domain, high-bit hardening reduces to a finite substitution that large rule sets can express. \chain{ESC}'s length expansion defeats it. Length-altering primitives can also disrupt hashcat's fixed-width kernels, but length-preserving \chain{BSL}/\chain{BSR} resist comparably, so structural grammar disruption is the primary defense.

\textbf{Robustness beyond rule files.} We also evaluated masks with Markov-based candidate ordering~\cite{narayanan2005fast} (hashcat \lit{-a 3}), \textsc{pcfg}~\cite{weir2009password}, and \textsc{prince}~\cite{steube2014prince} against $100$K digit-class targets on the Blackwell GPU using the guesser-specific procedures of \S\ref{sec:methodology}. Table~\ref{tab:nonrule_guessers} reports their partial recovery of unhardened targets and zero recovery of escape-hardened targets. The escape outputs non-printable bytes (\lit{0x80}--\lit{0x8F}), which these stock printable-only guessers do not emit. An attacker can keep each guesser's ranking over base plaintexts and apply candidate chains before hashing. The zero-recovery result therefore measures the reach of unmodified tools, not hardness against the defined \textsc{aware} adversary (\S\ref{sec:results_aware}).

\subsection{Performance on Resource-Constrained Hardware}
\label{sec:results_perf}

\begin{table*}[!t]
\renewcommand{\arraystretch}{1.15}
\caption{Direct ARMOR latency across IoT-class platforms and input lengths.}
\label{tab:embedded_latency}
\centering
\footnotesize
\setlength{\tabcolsep}{3pt}
\begin{tabular}{llrrrrrr}
\toprule
\textbf{Platform} & \textbf{Length} & \textbf{BSL} & \textbf{T3-HC1} & \textbf{T3-HC2} & \textbf{T3-HC3} & \textbf{T4-HC1} & \textbf{ESC-HC1} \\
\midrule
\multirow{4}{*}{\shortstack[l]{Cortex-M4\\@ 80\,MHz\\\scriptsize(NUCLEO-L432KC)}}
  &  6 & $93$ c ($1.16\,\mu$s) & $218$ c ($2.73\,\mu$s) & $354$ c ($4.43\,\mu$s) & $352$ c ($4.40\,\mu$s) & $473$ c ($5.91\,\mu$s) & $419$ c ($5.24\,\mu$s) \\
  & 11 & $153$ c ($1.91\,\mu$s) & $335$ c ($4.19\,\mu$s) & $579$ c ($7.24\,\mu$s) & $585$ c ($7.31\,\mu$s) & $768$ c ($9.60\,\mu$s) & $702$ c ($8.78\,\mu$s) \\
  & 14 & $189$ c ($2.36\,\mu$s) & $414$ c ($5.18\,\mu$s) & $714$ c ($8.93\,\mu$s) & $717$ c ($8.96\,\mu$s) & $945$ c ($11.81\,\mu$s) & $864$ c ($10.80\,\mu$s) \\
  & 20 & $261$ c ($3.26\,\mu$s) & $561$ c ($7.01\,\mu$s) & $984$ c ($12.30\,\mu$s) & $1{,}029$ c ($12.86\,\mu$s) & $1{,}299$ c ($16.24\,\mu$s) & $1{,}236$ c ($15.45\,\mu$s) \\
\midrule
\multirow{4}{*}{\shortstack[l]{Cortex-M33\\@ 150\,MHz\\(Pico~2W)}}
  &  6 & $89$ c ($0.59\,\mu$s) & $202$ c ($1.35\,\mu$s) & $334$ c ($2.23\,\mu$s) & $330$ c ($2.20\,\mu$s) & $440$ c ($2.93\,\mu$s) & $342$ c ($2.28\,\mu$s) \\
  & 11 & $149$ c ($0.99\,\mu$s) & $307$ c ($2.05\,\mu$s) & $552$ c ($3.68\,\mu$s) & $554$ c ($3.69\,\mu$s) & $725$ c ($4.83\,\mu$s) & $571$ c ($3.81\,\mu$s) \\
  & 14 & $185$ c ($1.23\,\mu$s) & $378$ c ($2.52\,\mu$s) & $678$ c ($4.52\,\mu$s) & $678$ c ($4.52\,\mu$s) & $896$ c ($5.97\,\mu$s) & $698$ c ($4.65\,\mu$s) \\
  & 20 & $257$ c ($1.71\,\mu$s) & $510$ c ($3.40\,\mu$s) & $930$ c ($6.20\,\mu$s) & $975$ c ($6.50\,\mu$s) & $1{,}246$ c ($8.31\,\mu$s) & $1{,}001$ c ($6.67\,\mu$s) \\
\midrule
\multirow{4}{*}{\shortstack[l]{Xtensa LX6\\@ 240\,MHz\\\scriptsize (ESP32 DevKit)}}
  &  6 & $105$ c ($0.44\,\mu$s) & $240$ c ($1.00\,\mu$s) & $335$ c ($1.40\,\mu$s) & $411$ c ($1.71\,\mu$s) & $532$ c ($2.22\,\mu$s) & $455$ c ($1.90\,\mu$s) \\
  & 11 & $175$ c ($0.73\,\mu$s) & $403$ c ($1.68\,\mu$s) & $560$ c ($2.33\,\mu$s) & $688$ c ($2.87\,\mu$s) & $897$ c ($3.74\,\mu$s) & $757$ c ($3.15\,\mu$s) \\
  & 14 & $217$ c ($0.90\,\mu$s) & $512$ c ($2.13\,\mu$s) & $695$ c ($2.90\,\mu$s) & $836$ c ($3.48\,\mu$s) & $1{,}116$ c ($4.65\,\mu$s) & $920$ c ($3.83\,\mu$s) \\
  & 20 & $301$ c ($1.25\,\mu$s) & $716$ c ($2.98\,\mu$s) & $965$ c ($4.02\,\mu$s) & $1{,}229$ c ($5.12\,\mu$s) & $1{,}554$ c ($6.48\,\mu$s) & $1{,}343$ c ($5.60\,\mu$s) \\
\midrule
\multicolumn{8}{>{\raggedright\arraybackslash}p{\dimexpr\linewidth-2\tabcolsep\relax}}{\scriptsize \textit{Note:} Cycle counts are $N{=}100$ minima after a $100$-iteration warmup, with interrupts masked; microseconds use the listed clock. ARM DWT counters were used on Cortex-M4/M33 and Xtensa CCOUNT on ESP32. Chains: T3-HC1=\texttt{ROT-R$\to$REV$\to$BSL-X}, T3-HC2=\texttt{DUP$\to$BSR-X$\to$ROT-L}, T3-HC3=\texttt{TOG$\to$RFL$\to$BSL-X}, T4-HC1=\texttt{ROT-L$\to$BSR-X$\to$RFL$\to$TOG}, ESC-HC1=\texttt{TOG$\to$DUP$\to$ESC}. ESC-HC1 remains below T4-HC1 on every board.}\\
\bottomrule
\end{tabular}
  \vspace{-6mm}
\end{table*}

Fig.~\ref{fig:security_cost}(a) reports per-operation ARMOR latency on the Pi~3 Linux reference. A representative T3 chain completes in $\approx 0.3$--$0.6\,\mu$s. For \chain{BSL-X}, overhead scales linearly at $\approx 8$\,ns per character (Table~\ref{tab:armor_overhead_S}), matching the per-primitive $O(|P|)$ term. Composing $k$ primitives gives $O(k \cdot |P|)$ chain complexity. Every evaluated chain stays below $0.7\,\mu$s across all length bands (Fig.~\ref{fig:security_cost}(b)), below one \lit{MD5} hash ($\approx 1.46\,\mu$s) and orders of magnitude below memory-hard KDF costs (Fig.~\ref{fig:security_cost}(c)). Latency alone does not provide offline resistance. The per-account conditional cost, reported as $G_k$ (Table~\ref{tab:guessing_number}) rather than the nominal $|\mathcal{O}|^k$ chain count, requires CR isolation (DR-4, \S\ref{sec:security_privilege}).

\textbf{Direct measurement across three IoT-class platforms.} Bare-metal measurements confirm ARMOR remains low-latency beyond the Pi~$3$: the heaviest T3 chain stays $\leq 12.86\,\mu$s across the three boards of \S\ref{sec:iot_edge_platforms} (Cortex-M4, Cortex-M33, and Xtensa~LX6). Table~\ref{tab:embedded_latency} reports per-platform minimum cycle counts and derived latencies across all four input lengths for the hardened chains (\lit{T3-HC1}--\lit{T3-HC3}, \lit{T4-HC1}) and the production \lit{ESC-HC1}. On the Cortex-M33, the $20$-byte \lit{T3-HC1} and \lit{T3-HC3} observations are $3.40\,\mu$s and $6.50\,\mu$s, respectively. A heap-allocation-ceiling test confirms the expected \lit{Argon2id} infeasibility: all boards top out at $\leq 256$\,KiB allocatable heap, below the evaluated OWASP-recommended $(19\,\mathrm{MiB},2,1)$ profile (\S\ref{sec:kdf_infeasibility}).

\textbf{Comparison to per-user peppering.} As a cryptographic-overhead baseline, we measured one HMAC-\lit{SHA-256} on the same three boards using portable software \lit{SHA-256}, validated against RFC~4231 and timed with the same interrupt-masked procedure. One HMAC over a short ($\leq 20$-byte) password costs $\approx 24{,}000$ cycles on the Cortex-M33 ($\approx 160\,\mu$s at $150$\,MHz). The Cortex-M4 and ESP32 require $\approx 337\,\mu$s and $\approx 125\,\mu$s, respectively. Across boards, HMAC is $\approx 19$--$21\times$ slower than ARMOR's evaluated \lit{T4-HC1} chain and $\approx 24$--$26\times$ slower than the slowest evaluated T3 chain (Table~\ref{tab:embedded_latency}). ARMOR is lighter in these pre-storage compute measurements and uses compact per-account chain state rather than high-entropy per-user pepper keys (\S\ref{sec:peppering}). These HMAC measurements use software \lit{SHA-256} for cross-platform parity. An on-chip SHA accelerator may narrow the latency gap. Peppering still requires protected key storage and a keyed cryptographic operation, while ARMOR requires bytewise preprocessing and CR isolation.

\subsection{Secure-Service Cost: ARMOR versus Peppering}
\label{sec:results_secure_service}

\rev{Table~\ref{tab:secure_service} reports separately timed compute, secure-storage, and boundary components for the two secure services of \S\ref{sec:secure_service_methodology}. Both use one Pico~2W TF-M image, with the CR held in ITS under dummy provisioning. The table sums measured component minima for each comparator, not complete authentication times.}

\begin{table}[!t]
\renewcommand{\arraystretch}{1.15}
\caption{\rev{Measured secure-service components and their summed minima on Pico~2W at $150$\,MHz under TF-M Isolation~$2$. The ARMOR component sum uses the validator-accepted maximum-expansion $k{=}5$ stress case.}}
\label{tab:secure_service}
\centering
\footnotesize
\setlength{\tabcolsep}{4pt}
\begin{tabular}{@{}lrr@{}}
\toprule
\textbf{Operation} & \textbf{Cycles} & \textbf{Latency} \\
\midrule
\multicolumn{3}{@{}l}{\textit{Keyed pre-hash compute}} \\
ARMOR, $k{=}1$ (single operation)            & $3{,}816$   & $25.4\,\mu$s \\
ARMOR, $k{=}5$ (max.\ expansion)     & $6{,}978$   & $46.5\,\mu$s \\
HMAC pepper, key held live           & $62{,}170$  & $414.5\,\mu$s \\
HMAC pepper, per-user key from ITS   & $125{,}775$ & $838.5\,\mu$s \\
\midrule
\multicolumn{3}{@{}l}{\textit{Secure storage and boundary}} \\
Chain-registry lookup (ARMOR)        & $279{,}901$ & $1{,}866.0\,\mu$s \\
Pepper key lookup (per-user)         & $285{,}080$ & $1{,}900.5\,\mu$s \\
Boundary crossing                    & $8{,}633$   & $57.6\,\mu$s \\
\midrule
\multicolumn{3}{@{}l}{\textit{Measured-component sums, ARMOR $k{=}5$}} \\
ARMOR                                & $295{,}512$ & $1{,}970.1\,\mu$s \\
Per-user pepper                      & $419{,}488$ & $2{,}796.6\,\mu$s \\
Global pepper, key cached            & $70{,}803$  & $472.0\,\mu$s \\
\bottomrule
\addlinespace[2pt]
\multicolumn{3}{@{}>{\raggedright\arraybackslash}p{0.95\columnwidth}@{}}{\scriptsize \textit{Note:} Values are minima of $256$ calls after warmup with a $15$-character password. Sums combine separately timed compute, applicable storage, and a no-op boundary estimate, not complete authentication times. ARMOR excludes chain validation, whereas per-user pepper compute includes key import. See \S\ref{sec:secure_service_methodology} for the procedure and full exclusions.}\\
\end{tabular}
\vspace{-3mm}
\end{table}

\rev{The ARMOR stress case is the validator-accepted chain \lit{REV}$\to$\lit{ROT-L}$\to$\lit{DUP}$\to$\lit{DUP}$\to$\lit{DUP}, followed by terminal \lit{ESC}. It exercises the production limit of three length-expanding operations without collapsing to a shorter chain. This selection tests maximum expansion, not a security-optimal depth or the worst execution time across all chains and inputs. For the $15$-character input, the minimum measured transform time within the secure service is $6{,}978$\,c. The secure-service and bare-metal measurements use different firmware builds and execution environments. Their transform times are therefore not interchangeable, and the difference cannot be attributed solely to isolation overhead.}

\rev{With its key imported during setup before timing, the PSA Crypto HMAC call has a $62{,}170$\,c minimum, or $8.91\times$ the measured ARMOR compute minimum. This ratio compares only the separately timed compute components. The PSA call includes framework work as well as HMAC computation, unlike the bare-metal software-HMAC measurement of \S\ref{sec:results_perf} ($\approx 24{,}000$\,c). It is distinct from the per-user ITS-read/import path, and no framework-overhead breakdown was measured. This build uses software \lit{SHA-256}. Hardware-accelerated HMAC was not evaluated, so the observed ratio does not establish a hardware-independent advantage.}

\rev{\textbf{Measured service components.} ARMOR reads a per-account chain from secure storage on every authentication. A per-user pepper also requires a per-account read, while the global-pepper comparator keeps one key live. These different storage costs dominate the measured sums.}

\rev{The per-account storage reads are similar at $279{,}901$\,c for ARMOR and $285{,}080$\,c for the per-user pepper, within $1.9\%$. Summing their measured service components gives $295{,}512$\,c for ARMOR and $419{,}488$\,c for the per-user pepper, so the per-user-pepper sum is $1.42\times$ ARMOR's. The cached global pepper has a $70{,}803$\,c component sum. ARMOR's is $4.17\times$ that sum. These sums combine separate minima for a $15$-character input. They are not complete authentication times.}

\rev{For the evaluated $k{=}5$ case, the ITS lookup required approximately $40\times$ as many cycles as the transform. Protected storage, rather than bytewise computation, dominates this service comparison. Caching per-account state would change the measured trade-off and was not evaluated here.}

\rev{Protected high-entropy per-user peppering provides stronger cryptographic protection under DB-only compromise. ARMOR instead combines compact protected per-account chain state with bytewise preprocessing and preserves the credential-record format. It targets resource- and compatibility-constrained deployments by providing bounded chain uncertainty under CR isolation (\S\ref{sec:peppering}). The measurements show low transform cost in the evaluated implementation and identify protected-storage access as the dominant component of both per-account services.}

\iffull

\fi

\subsection{Operational Implications: Attribution and Duress}
\label{sec:attribution_duress}

ARMOR's per-user keyed chain suggests two unevaluated operational extensions beyond offline-cost amplification, both relying on verifier-side chain state and CR isolation.

\textbf{1. Leak attribution.} Let $K_u$ denote account $u$'s production chain. Its transform is the injective mapping $P\mapsto\chain{ESC}(T_{K_u}(P))$. An attacker who recovers and rediscloses the transformed candidate rather than the plaintext may leave a value the CR-holding defender can test against stored verifiers. For each matching account, the defender can invert the transform to obtain a plaintext candidate. Per-account injectivity gives one plaintext for a fixed account and chain, but it does not guarantee a unique originating account across the store. This is a possible attribution signal, not a measured tracing guarantee. An attacker who knows the chain can instead publish the plaintext and erase that signal. Quantifying ambiguity and evasion requires a dedicated evaluation (\S\ref{sec:conclusion}).

\textbf{2. Duress signaling.} The same per-account indirection can support a panic-password variant. \rev{The CR may hold a second, duress chain for an account, which authenticates normally but raises a silent alarm. This costs a second stored verifier, so it gives up DR-1's literal schema preservation. Accepting two plaintexts for one account is a policy choice, not a pre-hash equivalence, since each chain over $\mathcal{O}_{\mathrm{inj}}$ stays injective (Lemma~\ref{lem:production_injectivity}).} Admitting one additional accepted chain per account changes the per-account search space by a small constant factor within the validator-bounded library $|\mathcal{V}_k|$. Duress signaling is therefore a possible design corollary of ARMOR's keyed-chain indirection rather than a separately evaluated mechanism.

\subsection{Assessment against Design Requirements}
\label{sec:results_summary}
Table~\ref{tab:design_requirements} maps the five design requirements to their supporting analysis and evidence. The construction preserves the credential-record format for completed enrollments (DR-1) and provides fixed-chain injectivity with local chain-validation rules (DR-2). The accepted T4 degeneration limits local pruning without invalidating fixed-chain injectivity. The evaluation quantifies conditional unit-cost chain uncertainty (DR-3) and implementation costs on the tested platforms (DR-5). The chain-search benefit requires CR isolation (DR-4), which remains a deployment assumption beyond the evaluated prototype. Table~\ref{tab:summary} summarizes selected experimental findings.

\section{Conclusion}
\label{sec:conclusion}

ARMOR adds a finite per-account chain-search subproblem under DB-only compromise, provided the CR remains isolated; CR exposure removes this chain-search benefit. The production construction preserves the credential-record schema for completed enrollments and uses injective operations followed by a fixed terminal escape encoding, preventing transform-induced pre-hash equivalence under each assigned chain. Its conditional search model additionally requires independent, uniform chain assignment over the accepted production set. Although the validator constrains chain selection, the accepted T4 degeneration demonstrates that adjacent-pair validation does not exclude every algebraic collapse.

The attack experiments distinguish stock-tool candidate coverage from chain-aware search. Unhardened chains leave a digit-only bypass, with stock-tool recovery reaching $45.3\%$, whereas no targets are recovered in the evaluated $100{,}000$-target escape-hardened digit-only configurations. This result characterizes the tested tools' candidate coverage. Exhaustive historical-library enumeration, in turn, matches every transformed target at full dictionary coverage, but an arbitrary-chain match need not identify a login-valid password.

For production chains, the median conditional guessing number over the eligible RockYou corpus is approximately $52$ expected output-class tests at depth $k=3$ and $317$ at $k=5$. These values assume the correct base password and unit cost per class test; they do not measure the cost of jointly searching unknown passwords and chains across salted accounts. The separate bounded production attack executes $89.35\times$ the plain arm's salted checks for $1{,}000$ paired targets, with identical recovery ranks on A6000 and A100. Reused inverse enumeration and posterior-ordered simulated submissions recover all original passwords within that pool. These findings characterize the implemented attack and do not establish a universal minimum work factor.

On the defender side, implementation measurements support low transform cost on constrained hardware, while showing that secure storage can dominate deployment overhead. On the Pico~2W, the measured component minima sum to approximately $1.97$\,ms for the depth-$5$ stress case. This sum excludes per-call chain validation and the DB lookup, final hashing, and comparison needed for complete authentication. ARMOR therefore offers bounded chain uncertainty under explicit isolation and assignment conditions, while protected high-entropy peppers provide stronger cryptographic protection. KDF composition remains applicable when resources permit and the KDF accepts the complete transformed output.

\subsection{Future Work}

Future work will explore three extensions of ARMOR: duress-aware authentication, evaluating alarm reliability and resistance to misuse; forensic attribution of leaked transformed credentials, accounting for ambiguity and adversarial evasion; and Chain Registry lifecycle management, including secure chain rotation and recovery. These directions extend ARMOR's per-account chain structure beyond the offline-guessing setting evaluated here.

\appendices
\raggedbottom
\setlength{\subfigtopskip}{0pt}
\setlength{\subfigbottomskip}{0pt}
\setlength{\subfigcapskip}{2pt}
\section{Chain-Validator Enumeration}\label{app:I}
The production enrollment validator accepts only chains that avoid the mechanically detectable adjacent degeneracies stated in the main paper. It rejects an adjacent cross-primitive inverse pair, specifically \lit{(BSL, BSR)} or \lit{(ROT-L, ROT-R)}, and an adjacent repeated involution, \chain{TOG}\,\chain{TOG} or \chain{REV}\,\chain{REV}. This local rule does not eliminate every non-adjacent algebraic equivalence. The $G_k$ calculation groups chains with the same output into equivalence classes. Let $e(K)$ count the length-doubling \chain{DUP} and \chain{RFL} operations in chain $K$, as in \S~IV-E of the main paper. Production enrollment also requires $e(K)\le e_{\max}=3$. Let $\mathcal{V}_k\subseteq\mathcal{O}_{\mathrm{inj}}^{\,k}$ denote the chains accepted by both checks. These chains contain no adjacent cancellation-inverse pair, no adjacent repeated involution, and at most three length-doubling operations. Assignment checks take $O(k)$ at enrollment. The Trusted Firmware-M (TF-M) prototype also revalidates each retrieved chain during authentication. Every $|\mathcal{V}_k|$ reported here (Table~\ref{tab:e3_distribution}, Appendix~\ref{app:VI}) is the size of this jointly constrained set. The length bound is inactive for $k\le 3$ and removes further chains at $k{=}4$ and $k{=}5$.

The historical transformed-target coverage evaluation is not a second production validator. It exhaustively enumerates the complete $|\mathcal{O}|{=}14$ evaluation library, including non-injective and case-normalizing operations, with $14^3=2{,}744$ orderings at depth three. Its coverage does not depend on production-validator pruning. A match under any chain does not establish that the candidate authenticates under the assigned chain. Production chain counts and $G_k$ use only the eight-operation set $\mathcal{O}_{\mathrm{inj}}$ and the enrollment validator above, which applies both the adjacency rules and the length bound (Table~\ref{tab:e3_distribution}).

Definition~V.1 and Proposition~2 in \S~V-F of the main paper define $G_k(p)$ for the accepted production set $\mathcal{V}_k$ and prove its bounds. Table~V of the main paper reports the account-weighted median of $\log_2 G_k(p)$ in bits and its corresponding linear factor.

\section{Injective Escape-Encoding Details}\label{app:II}
A production hardening primitive must satisfy two requirements simultaneously. First, because the verifier checks a digest over $\textsc{armor}(P,K)$ and the account salt, the transform must be injective. Otherwise, distinct plaintexts can authenticate to the same account. Second, closing the digit-only bypass requires the transformed output to leave the printable-ASCII grammar modeled by rule-based password crackers.

The evaluated high-bit byte maps satisfy only one requirement at a time. The high-bit OR form used in \chain{BSL-X}/\chain{BSR-X} maps printable inputs outside standard printable ASCII, but it is $2$-to-$1$ and can create authentication collisions. The alternative \chain{BSL-XOR}/\chain{BSR-XOR} form ($\oplus\,\lit{0x80}$) is injective, but it does not universally escape: some non-digit inputs can re-enter printable ASCII and raise \textsc{naive} recovery in the hardening comparison (Table~VIII of the main paper). This limitation is intrinsic to byte-preserving maps, since a byte bijection cannot map all $256$ input values into the $161$ non-printable byte values.

The escape primitive \chain{ESC} resolves the conflict by expanding each byte into two non-printable bytes:
\[
b\mapsto(\lit{0x80}\,|\,b{\gg}4,\ \lit{0x80}\,|\,b\,\&\,\lit{0xf}).
\]
The mapping is injective because the two nibbles recombine to recover $b$, and it is universally escaping because both output bytes lie in the non-printable range used by ARMOR's byte-level disruption. This doubles the pre-hash byte length via a lookup-style byte expansion, but the stored verifier remains a fixed-size digest, so the credential schema is unchanged. The specified production construction uses $\mathcal{O}_{\mathrm{inj}}$ and appends a fixed terminal \chain{ESC} to every accepted chain. In the main results, escape-ending chains yield no recovered targets in the evaluated \textsc{naive} attacks across all rule files, digit and non-digit targets, and dictionaries up to $10^6$ words.

\section{Authentication-Equivalence Details}\label{app:III}

The non-injective variants in the evaluation library require careful interpretation. The $1$M sampled RockYou entries contain $D=999{,}997$ distinct plaintexts. For a fixed chain with $U$ distinct transformed values, the distinct-value loss is $(D-U)/D$. It is $\le 0.08\%$ for the near-injective \chain{BSL-X} chains and $3.99\%$ for \chain{DUP$\to$BSR-X$\to$ROT-L}. This is not the fraction of plaintexts belonging to shared-output classes. The main paper (\S V-D) distinguishes output collapse from per-account false-accept exposure. The production library removes this transform-induced equivalence by replacing non-injective high-bit-OR diffusion with injective \chain{ESC} and excluding case maps from $\mathcal{O}_{\mathrm{inj}}$. Every production chain is therefore injective. Deployments that deliberately tolerate case-insensitive login may retain case maps as a deployment policy, accepting bounded case-structured equivalence.

Lemma~1 in \S~V-D of the main paper establishes fixed-chain production injectivity and gives its proof. The result concerns distinct passwords under the same accepted chain. With the same salt, equal verifiers would require a collision of the underlying hash. It does not establish uniqueness across different chains. Appendix~\ref{app:II} gives the escape-encoding details used in the proof.

\section{Per-User Peppering Comparison Details}\label{app:IV}

The closest peppering analog is per-user peppering, which assigns each user a unique high-entropy secret and shares ARMOR's per-user isolation property. A protected 128-bit per-user \lit{pepper} contributes $2^{128}$ to the per-user search space. ARMOR contributes a bounded conditional chain subproblem. The nominal production count $|\mathcal{O}_{\mathrm{inj}}|^k=8^k$ is an upper bound, while the production result is reported by the smaller multiplicity-aware $G_k(p)$. Per-user peppering is therefore cryptographically stronger wherever it can be provisioned and managed correctly, but it is operationally heavier on RAM-constrained credential repositories.

ARMOR and peppering have different failure modes. ARMOR's bit-level primitives push ASCII inputs outside the character-space grammar modeled by hashcat rule sets. Per-user peppering adds secret material but does not change the password-mutation grammar. \rev{If an independently protected \lit{pepper} is exposed while the chain remains secret, the pepper benefit is lost but the conditional ARMOR chain subproblem remains. If the Chain Registry (CR) is exposed, an attacker with the account-to-chain mapping can apply each transform before hashing, removing ARMOR's chain-search benefit. The underlying verifier or key derivation function (KDF), and any independently protected pepper, retain their own protection. A mismatch with unmodified stock tools does not protect against an attacker who knows the chain.} An assigned ARMOR chain can be represented in $\lceil k\log_2|\mathcal{O}|\rceil$ bits using the $14$-operation evaluation library ($12$ bits at $k{=}3$). The production-theoretic minimum over the accepted set is $\lceil\log_2|\mathcal{V}_3|\rceil = 9$ bits at $k{=}3$. These are encoding capacities, not measured serialized-record sizes or security advantages. They are smaller than a high-entropy per-user \lit{pepper}, \rev{while the credential database (DB) stores no additional per-user state at all.}

As an embedded-cost reference, a post-hashing \lit{pepper} scheme based on a hash-based message authentication code (HMAC) adds a keyed-hash invocation~\cite{owaspPasswordStorageCheatSheet}. Across the evaluated embedded targets, software HMAC-\lit{SHA-256} is $\approx 19$--$21\times$ slower than ARMOR's worst-case \lit{T4-HC1} chain and $\approx 24$--$26\times$ slower than the slowest evaluated T3 chain. These timings are a footprint reference for crypto-accelerator-less microcontrollers, not a security comparison against peppering. Where both mechanisms are deployable, they can be layered as \lit{Hash}(\lit{pepper}$\|$\lit{ARMOR(password)}) or \lit{KDF}(\lit{pepper}$\|$\lit{ARMOR(password)}), provided the KDF accepts the full combined input.

\iffull\else
\begin{table}[!ht]
\centering
\renewcommand{\arraystretch}{1.2}
\caption{\small Evaluation and production chain-space quantities.}
\label{tab:e3_distribution}
\footnotesize
\begin{tabular}{@{}lr@{}}
\toprule
\textbf{Quantity} & \textbf{Value} \\
\midrule
Evaluation library, nominal depth-$3$ chains & $14^3=2{,}744$ \\
Production library, nominal depth-$3$ chains & $8^3=512$ \\
Validator-accepted production chains, $|\mathcal{V}_3|$ & $422$ \\
Account-weighted median production $G_3$ & $52.2$ class tests \\
\midrule
\multicolumn{2}{>{\raggedright\arraybackslash}p{\dimexpr\linewidth-2\tabcolsep\relax}}{\scriptsize \textit{Note.} Production uses eight injective non-terminal operations, both enrollment checks in Appendix~\ref{app:I}, and terminal \chain{ESC}. $G_3$ is conditional on the correct base plaintext and is not an end-to-end cracking multiplier.}\\
\bottomrule
\end{tabular}
\end{table}

\begin{table}[!ht]
\centering
\renewcommand{\arraystretch}{1.2}
\caption{\small Pi~3 ARMOR primitive cost for 20-character inputs.}
\label{tab:armor_overhead_S}
\footnotesize
\setlength{\tabcolsep}{2pt}

\begin{tabular}{>{\raggedright\arraybackslash}p{0.49\linewidth}|r|c}
\toprule
\textbf{Category} & \textbf{Latency ($\mu$s)} & \textbf{\% of one MD5} \\
\midrule
Structural (\texttt{REV}, \texttt{ROT}, \texttt{DUP}, \texttt{RFL}) & $0.06$--$0.08\,\mu$s & $4$--$6\%$ \\
Case maps (\texttt{L/U-CASE}, \texttt{CAP}, \texttt{CAP-INV})$^{\ast}$ & $0.14$--$0.15\,\mu$s & $\approx 10\%$ \\
Case toggle (\texttt{TOG})$^{\dagger}$       & $0.25\,\mu$s         & $17\%$ \\
\midrule
Rotate: \texttt{BSL-1} / \texttt{BSR-1}      & $0.15\,\mu$s         & $11\%$ \\
Hardened: \texttt{BSL-X} / \texttt{BSR-X}  & $0.16\,\mu$s         & $11\%$ \\
\midrule
\multicolumn{3}{>{\raggedright\arraybackslash}p{\dimexpr\linewidth-2\tabcolsep\relax}}{\scriptsize \textit{Note.} Minimum of $4000$-call batches after baseline subtraction. BSL-X/BSR-X differs from BSL-1/BSR-1 by less than measurement noise. $^{\ast}$Non-injective and evaluation-only. $^{\dagger}$Injective and retained in $\mathcal{O}_{\mathrm{inj}}$.}\\
\bottomrule
\end{tabular}
\end{table}
\fi

\section{Tables}\label{app:V}
Tables~\ref{tab:iot_edge_platforms}--\ref{tab:summary} provide the additional platform, attack-coverage, and design-evidence details cited in the main text.
\begin{table}[!ht]
\begingroup
\caption{IoT and edge evaluation platforms.}
\label{tab:iot_edge_platforms}
\centering
\footnotesize
\setlength{\tabcolsep}{3pt}
\begin{tabularx}{\linewidth}{@{}>{\raggedright\arraybackslash}p{0.23\linewidth}X@{}}
\toprule
\textbf{Platform} & \textbf{Role and hardware} \\
\midrule
Pi~$3$ Model~B & Linux-class Cortex-A53 edge-gateway reference; ARMv8 @ $1.2$\,GHz, $1$\,GB LPDDR2; only platform compared head-to-head with \chain{Argon2id}. \\
NUCLEO-L432KC & Bare-metal Cortex-M4; STM32L432 @ $80$\,MHz, $64$\,KiB SRAM, $256$\,KiB Flash; Arduino-STM32 firmware. \\
Pico~$2$W & Bare-metal Cortex-M33 with TrustZone-M; RP2350 @ $150$\,MHz, $520$\,KiB SRAM, $4$\,MiB Flash; Arduino-Pico firmware. \\
ESP32 DevKit & Non-ARM IoT cross-architecture point; Xtensa LX6 @ $240$\,MHz, $520$\,KiB SRAM, $4$\,MiB Flash. \\
\bottomrule
\end{tabularx}
\endgroup\par\medskip
\begingroup
\centering
\caption{RockYou digit-only prevalence.}
\label{tab:digit_prevalence}
\renewcommand{\arraystretch}{1.2}
\footnotesize
\setlength{\tabcolsep}{3pt}
\begin{tabular}{@{}rrrr@{}}
\toprule
\textbf{Top-N window} & \textbf{Digit-only} & \textbf{Total} & \textbf{\% Digit-only} \\
\midrule
$100$           & $11$        & $100$           & $11.00\%$ \\
$500$           & $41$        & $500$           & $8.20\%$ \\
$2{,}000$       & $100$       & $2{,}000$       & $5.00\%$ \\
$10{,}000$      & $427$       & $10{,}000$      & $4.27\%$ \\
$100{,}000$     & $19{,}052$  & $100{,}000$     & $19.05\%$ \\
$1{,}000{,}000$ & $115{,}731$ & $1{,}000{,}000$ & $11.57\%$ \\
\bottomrule
\end{tabular}\par
\endgroup
\end{table}

\begin{table*}[!htp]
\centering
\caption{\small ARMOR and KDF deployment tradeoffs.}
\label{tab:ARMOR_vs_kdf_comprehensive}
\footnotesize
\setlength{\tabcolsep}{3pt}
\begin{tabular}{@{}>{\raggedright\arraybackslash}p{0.15\textwidth}>{\raggedright\arraybackslash}p{0.408\textwidth}>{\raggedright\arraybackslash}p{0.332\textwidth}@{}}
\toprule
\textbf{Aspect} & \textbf{ARMOR} & \textbf{KDFs (\tablelit{Argon2}, \tablelit{scrypt}, \tablelit{bcrypt})} \\
\midrule

\textbf{Primary defense} & Adds a hidden per-user transform and a bounded conditional output-class search subproblem when the CR remains isolated. Its measured unit-cost magnitude is $G_k(p)$, not the nominal $|\mathcal{O}|^k$. & Deters brute-force attacks by making each hash computation prohibitively slow and resource-intensive. \\
\addlinespace[0.5mm]
\hline

\textbf{Core mechanism} & Deterministic per-user keyed transformation chain applied to plaintext before hashing. & Deliberately slow, iterative algorithm with configurable work factor and memory cost. \\
\addlinespace[0.5mm]
\hline

\textbf{Integration} & Pre-hash step at enrollment and authentication. Use \tablelit{KDF(ARMOR(password))} only where the full output fits the KDF input limit and \tablelit{Hash(ARMOR(password))} on constrained devices. & Terminal function that replaces the hash primitive directly. \\
\addlinespace[0.5mm]
\hline

\textbf{Computational cost} & Fixed and low: $O(k \cdot |P|)$ byte/string work per authentication. & Configurable and high, tunable to remain resistant as hardware improves. \\
\addlinespace[0.5mm]
\hline

\textbf{Memory usage} & Compact chain state plus bounded working buffers; tested on sub-megabyte embedded hardware. & Tens of MiB or more for recommended memory-hard settings (e.g., the $64$\,MiB RFC~9106 Argon2id profile), infeasible at security-meaningful settings on Cortex-M class devices~\cite{rfc9106}. \\
\addlinespace[0.5mm]
\hline

\textbf{Resistance type} & Operational cost amplification under Kerckhoffs' Principle~\cite{Petitcolas2025}. No cryptographic-hardness claim. & Cryptographic-strength resistance proportional to the configured work factor. \\
\bottomrule
\end{tabular}
\end{table*}

\iffull\else

\begin{table*}[!htp]
\centering
\begin{minipage}[t]{0.45\textwidth}
\centering
\renewcommand{\arraystretch}{1.05}
\caption{\small Digit-only T3 recovery by rule file.}
\label{tab:digit_perfile_S}
\scriptsize
\setlength{\tabcolsep}{2pt}
\begin{tabular}{@{}l r c@{}}
\toprule
\textbf{Rule file} & \textbf{NH \%} & \textbf{ESC \%} \\
\midrule
\texttt{SuperUnicorn} & $45.27$ & $0.000$ \\
\texttt{Pantacorn} & $40.42$ & $0.000$ \\
\texttt{UnicornLorge} & $36.16$ & $0.000$ \\
\texttt{omgwtfbbq} & $34.52$ & $0.000$ \\
\texttt{sapphire\_v1} & $34.11$ & $0.000$ \\
\texttt{buka\_400k} & $32.38$ & $0.000$ \\
\texttt{concentrator\_MT\_250000} & $27.80$ & $0.000$ \\
\texttt{sapphire\_v3} & $27.01$ & $0.000$ \\
\texttt{collection500k} & $26.58$ & $0.000$ \\
\texttt{HashMob.150k} & $19.88$ & $0.000$ \\
\texttt{Unicorn20k} & $9.87$ & $0.000$ \\
\texttt{Unicorn10k} & $7.09$ & $0.000$ \\
\texttt{Unicorn5k} & $5.25$ & $0.000$ \\
\texttt{Unicorn3k} & $4.43$ & $0.000$ \\
\texttt{Unicorn1k} & $2.56$ & $0.000$ \\
\texttt{best64} & $1.11$ & $0.000$ \\
\midrule
\multicolumn{3}{>{\raggedright\arraybackslash}p{\dimexpr\linewidth-2\tabcolsep\relax}}{\scriptsize \textit{Note.} \textsc{naive} recovery using hashcat on $100$K digit-only targets with a dictionary to $10^6$ on Blackwell. \textbf{NH} denotes the non-hardened T3 chain \chain{ROT-R}$\to$\chain{REV}$\to$\chain{BSL}. \textbf{ESC} denotes production escape.}\\
\bottomrule
\end{tabular}
\end{minipage}\hfill
\begin{minipage}[t]{0.54\textwidth}
\renewcommand{\arraystretch}{1.08}
\caption{\small Offline-evaluation coverage.}
\label{tab:eval_scale_S}
\raggedright
\scriptsize
\setlength{\tabcolsep}{3pt}
\begin{tabular}{@{}>{\raggedright\arraybackslash}p{0.26\linewidth}>{\raggedright\arraybackslash}p{0.20\linewidth}>{\raggedright\arraybackslash}p{0.44\linewidth}@{}}
\toprule
\textbf{Evaluation block} & \textbf{Platform} & \textbf{Coverage} \\
\midrule
Main \textsc{naive} evaluation & A100 & Rule, dictionary, and brute-force attacks over the tier configurations; four hash algorithms; fixed time budgets of $30$/$60$/$300$/$600$\,s \\
Extended \textit{OMG} evaluation & A100 & $21$ configurations $\times$ $4$ hashes $\times$ $3$ target sizes, at four fixed time budgets \\
Hardened evaluation & A100 + Blackwell & Hardened configurations plus baseline \\
\textsc{aware} chain-expanded evaluation & CPU-exact & $3$ hardened T3 chains $\times$ $4$ dictionary prefixes, fixed $1$M-account targets \\
\midrule
Digit-class hardening study & Blackwell & Escape-encoding; $16$ rule files; digit and non-digit targets; dictionary to $10^6$ \\
Non-rule guessers & Blackwell & Markov, \textsc{prince}, and \textsc{pcfg} on plain and escape digit-class targets \\
\midrule
\multicolumn{3}{>{\raggedright\arraybackslash}p{\dimexpr\linewidth-2\tabcolsep\relax}}{\scriptsize \textit{Note.} The blocks use different counting units, so no aggregate experiment count is defined. RockYou-derived target sets use seed $s=42$. PCFG denotes probabilistic context-free grammar. \textsc{aware} runs continue to candidate-stream exhaustion.}\\
\bottomrule
\end{tabular}
\end{minipage}
\end{table*}

\fi

\iffull\else

\begin{table*}[!tp]
\centering
\begin{minipage}[t]{0.53\textwidth}
\renewcommand{\arraystretch}{1.2}
\caption{\small Per-chain \textsc{naive} crack rates for selected T2 and T4 chains (unweighted means across evaluated rule-based configurations, A100).}
\label{tab:perchain_inverse}
\centering
\footnotesize
\setlength{\tabcolsep}{3pt}
\begin{tabular}{llc}
\toprule
\textbf{Tier / chain} & \textbf{Validator} & \textbf{Mean crack rate} \\
\midrule
\multicolumn{3}{@{}l}{\textit{T2 ($2$-op)}}\\
\quad\lit{T2-BC1}: \chain{BSL}$\to$\chain{ROT-L} & accepted & $\leq 1.5\%$ \\
\quad\lit{T2-BC2}: \chain{BSR}$\to$\chain{REV}   & accepted & $\leq 1.5\%$ \\
\quad\lit{T2-BC3}: \chain{BSL}$\to$\chain{BSR}   & rejected & $86.56\%$ \\
\midrule
\multicolumn{3}{@{}l}{\textit{T4 ($4$-op)}}\\
\quad\lit{T4-BC1}: \chain{ROT-L}$\to$\chain{BSR}$\to$\chain{RFL}$\to$\chain{TOG} & accepted & $\leq 0.10\%$ \\
\quad\lit{T4-BC2}: \chain{DUP}$\to$\chain{BSL}$\to$\chain{REV}$\to$\chain{ROT-L} & accepted & $\leq 0.10\%$ \\
\quad\lit{T4-BC3}: \chain{BSL}$\to$\chain{REV}$\to$\chain{BSR}$\to$\chain{DUP}   & accepted & $27.58\%$ \\
\midrule
\multicolumn{3}{>{\raggedright\arraybackslash}p{\dimexpr\linewidth-2\tabcolsep\relax}}{\scriptsize \textit{Note.} Each mean equally weights the $288$ rule-based configurations for that chain. Inequalities bound chain means, not individual tests. The adjacent inverse pair in \lit{T2-BC3} is rejected. \lit{T4-BC3} passes the local validator but has a non-adjacent algebraic degeneration that raises recovery.}\\
\bottomrule
\end{tabular}
\end{minipage}\hfill
\begin{minipage}[t]{0.45\textwidth}
\renewcommand{\arraystretch}{1.2}
\caption{Rule-file characteristics.}
\label{tab:rulefiles}
\centering
\footnotesize
\setlength{\tabcolsep}{2pt}
\begin{tabular}{>{\raggedright\arraybackslash}p{0.31\linewidth}lrrrr}
\hline
\multicolumn{1}{l}{\textbf{Rule File}} & \multicolumn{1}{l}{\textbf{Short}} & \multicolumn{1}{r}{\textbf{Rules}} & \multicolumn{1}{r}{\textbf{\lit{d}$\,\geq$}} & \multicolumn{1}{r}{\textbf{\lit{f}$\,\geq$}} & \multicolumn{1}{r}{\textbf{b-rot}} \\
\hline
best64       & \texttt{B64} & $64$            & $2$       & $1$       & $0$ \\
Unicorn10k   & \texttt{U10} & $10{,}000$      & $27$      & $1$       & $0$ \\
concentrator\_MT\_\allowbreak250000 & \texttt{CMT} & $250{,}000$     & $519$     & $103$     & $0$ \\
buka\_400k   & \texttt{BUK} & $399{,}986$     & $1{,}601$ & $169$     & $0$ \\
sapphire\_v3 & \texttt{S3}  & $749{,}997$     & $2{,}919$ & $621$     & $0$ \\
sapphire\_v1 & \texttt{S1}  & $1{,}270{,}145$ & $8{,}397$ & $4{,}743$ & $0$ \\
omgwtfbbq    & \texttt{OMG} & $6{,}724{,}974$ & $26{,}201$& $12{,}281$& $0$ \\
\midrule
\multicolumn{6}{@{}p{0.95\linewidth}@{}}{\scriptsize \textbf{Short}: code used in the figures. \lit{d}/\lit{f}: native DUP/RFL rule counts (lower bounds). \textbf{b-rot}: bit-rotation operators (all $0$; hashcat rules define none).} \\
\bottomrule
\end{tabular}
\vspace{-3mm}
\end{minipage}
\end{table*}

\fi



%

\begin{table*}[!htp]
\centering
\caption{Comparison of server-side password-storage mechanisms}
\label{tab:ARMOR_vs_all}
\scriptsize
\setlength{\tabcolsep}{3pt}
\begin{tabular}{llllll}
\toprule
\textbf{Mechanism} & \textbf{Primary Goal} & \textbf{Nature} & \textbf{Scope} & \textbf{Memory Cost} & \textbf{Embedded Viable} \\
\midrule
Salting   & Uniqueness            & Random, non-secret          & Per-password  & Negligible  & Yes \\
Peppering & Obfuscation           & Secret key, static          & System-wide   & Negligible  & Yes \\
Memory-hard KDF & Computational cost    & Slow algorithm              & Per-password  & Profile-dependent& No for evaluated $19$\,MiB profile  \\
ARMOR     & Adversary cost amplification & Per-user keyed transformation & Per-user& Bounded  & Yes \\
\bottomrule
\end{tabular}
\end{table*}

\begin{table*}[!htp]
\caption{Threat scope of ARMOR.}
\label{tab:threat_scope}
\centering
\footnotesize
\begin{tabularx}{\textwidth}{@{}>{\raggedright\arraybackslash}p{0.075\textwidth} >{\raggedright\arraybackslash}p{0.205\textwidth} >{\raggedright\arraybackslash}p{0.66\textwidth}@{}}
\toprule
\textbf{Scope} & \textbf{Threat Class} & \textbf{Rationale} \\
\midrule
\textbf{In} & \textit{Offline dictionary/rule-based cracking} & DB-only attacker lacks the CR mapping. An \textsc{aware} attacker enumerates hidden chains, while a \textsc{naive} attacker is disrupted by rule-grammar mismatch. The bounded digit-only case is analyzed in \S~VII-G. \\\cline{2-3}
& \textit{Precomputation} & Transformed candidates can be generated once per plaintext and reused across accounts. Account salts prevent reuse of precomputed verifier hashes across accounts. Device or gateway binding may further limit transformed-candidate reuse. \\
\midrule
\textbf{Out} & \textit{Online brute-force} & Requires network- and application-layer defenses such as rate limiting and lockout policies. \\\cline{2-3}
& \textit{Credential disclosure} & Phishing, social engineering, and coercion obtain the plaintext directly, bypassing stored-verifier protection. \\\cline{2-3}
& \textit{Compromise of CR} & Exposure of \rev{$\lit{user\_id}\mapsto K$} removes ARMOR's chain-search benefit. The attacker still faces the underlying hash or KDF and any independently protected pepper. \\\cline{2-3}
& \textit{Client-side malware} & Malware captures credentials before any verifier-side protection can be applied. \\
\bottomrule
\end{tabularx}
\end{table*}

\begin{table*}[!tp]
\renewcommand{\arraystretch}{1.3}
\caption{ARMOR System Design Requirements}
\label{tab:design_requirements}
\centering
\footnotesize
\begin{tabularx}{\textwidth}{>{\bfseries}p{1.0cm}>{\raggedright\arraybackslash}X>{\raggedright\arraybackslash}p{2.2cm}}
\hline
\textbf{ID} & \textbf{Requirement} & \textbf{Evidence} \\
\hline
DR-1 & \textit{Schema Transparency:} Preserve the stored verifier format and require no protocol or client-interface changes; \rev{no state is added to the credential record}. & \rev{\S~IV-B3,~IV-D} \\
DR-2 & \textit{Functional Integrity \& Extensibility.} Reject the specified adjacent inverse pairs and repeated involutions at enrollment. Require injective production chains so no two plaintexts share an ARMOR pre-hash value under the same chain. & \S~IV-E,~IV-F \\
DR-3 & \textit{Adversary Cost Amplification:} \rev{For a correct base plaintext, impose a measured conditional unit-cost output-class guessing number $G_k$ (bounded above pointwise by the production distinct-output count $E^{\mathrm{prod}}_k(p)$ and the nominal production chain count $|\mathcal{O}_{\mathrm{inj}}|^k$); $G_k$ does not measure the joint unknown-plaintext and salted-check attack.} Defender work remains $O(k \cdot |P|)$; operational, not cryptographic, hardness. & \S~V-F,~VII-F \\
DR-4 & \textit{Separation of Compromise Domains:} Place the credential DB and CR in independent compromise domains so DB-only exfiltration does not reveal the per-user keyed chain. & \S~V-C \\
DR-5 & \textit{Resource Compatibility.} Keep the transform microsecond-scale on bare-metal Internet of Things (IoT) hardware and use bounded working memory. Report the measured secure-service component sums for a 15-character input separately from full authentication. & \S~VII-H,~VII-I \\
\hline
\end{tabularx}
\vspace{-1mm}
\end{table*}

\begin{table*}[!tp]
\renewcommand{\arraystretch}{1.1}
\caption{Summary of Key Experimental Findings}
\label{tab:summary}
\centering
\setlength{\tabcolsep}{2pt}
\footnotesize
\begin{tabular}{@{}p{0.34\textwidth}ccc}
\toprule
\textbf{Configuration} & \shortstack{\textbf{GPU} \textsc{naive}\\\scriptsize rule-based} & \shortstack{\rev{\textsc{aware} \textbf{coverage}}\\\scriptsize \rev{$D{=}10^6$, $14^3$ chains}} & \shortstack{\textbf{Verifier overhead}\\\scriptsize Pi~3 transform} \\
\midrule
Baseline                 & 86.52\%  & N/A & N/A \\
T0 (single-op)            & 57.92\%  & N/A & N/A \\
T1 (2-op structural)      & 33.49\%  & N/A & N/A \\
T2 (2-op, \chain{BSL}/\chain{BSR}) & 29.34\%  & N/A & N/A \\
\tablelit{T3-HC1}                    & $\leq$0.01\% & \rev{9.21\%}  & $0.28\,\mu$s \\
\tablelit{T3-HC2}                    & $\leq$0.01\% & \rev{14.62\%} & $0.42\,\mu$s \\
\tablelit{T3-HC3}                    & $\leq$0.01\% & \rev{9.53\%}  & $0.58\,\mu$s \\
\tablelit{T4-HC1}                    & $\leq$0.01\% & N/A & $0.62\,\mu$s \\
\midrule
\multicolumn{4}{@{}p{0.95\textwidth}@{}}{\scriptsize \textit{Note.} Graphics-processing-unit (GPU) \textsc{naive} values for Baseline and T0--T2 are unweighted means over the evaluated rule-based configurations. Hardened HC values are mixed-target maxima. \textsc{aware} values are account-weighted transformed-target coverage from main-paper Table IV, not login-valid password recovery. Digit-only results appear in main-paper Table~VIII.} \\
\bottomrule
\end{tabular}
\vspace{-4mm}
\end{table*}

\setlength{\dbltextfloatsep}{10pt plus 2pt minus 2pt}
\section{Enrollment-Leakage Simulation}\label{app:VI}

This section supports the chosen-password enrollment adversary of \S~V-B. The adversary controls
accounts $A$ ($|A|=m$) and targets a victim $v \notin A$. After a full breach of the verifier store,
it holds every salt and stored record. For its own accounts it also knows the plaintext it chose.
It can recover each of its own chains up to the equivalence class
$C_i=\{c \in \mathcal{V}_k : \textsc{armor}(P_i,c) = \textsc{armor}(P_i,K_i)\}$.
Let $E_A$ contain observations from the adversary's own enrollments. It excludes the victim's verifier even
though that verifier is part of the breached DB. The simulation gives the predictor the exact
chains $K_i$ assigned to the adversary's accounts rather than only their classes $C_i$. Its null result concerns leakage through
own-enrollment observations under independent assignment, not the entire breached-DB view.

\textbf{Method.} For each scheme and depth $k=1\ldots5$, we simulate $10^6$ enrollments, except
$2\times10^5$ for the linear congruential generator (\lit{lcg}) diagnostic. The corresponding
$|\mathcal{V}_k|$ values are $8$, $58$, $422$, $3054$, and $21822$. We use $10^5$ victims, adversary
sizes $m \in \{10^3,10^4,10^5\}$. The $\chi^2$ tests in T1--T3 use $\alpha=10^{-3}$. The specified sampler uses a
cryptographically secure pseudorandom number generator (CSPRNG). T1 tests marginal uniformity with
$\chi^2$ goodness of fit. T2 tests bucketed pair dependence and reports lag-$1$ through lag-$8$
autocorrelation. T3 tests account-identifier dependence with $\chi^2$ and reports mutual
information relative to a permutation null. T4 flags posterior top-$r$ hit rates more than three
estimated standard errors from chance. Its top-$1$ accuracy is descriptive. T5 describes
equivalence-class structure over $1000$ RockYou passwords and is not a leakage hypothesis test.
T4 compares two predictors.
\lit{freq} ranks candidate chains by their observed frequency across the adversary's accounts.
\lit{struct} fits an affine map from account identifier to chain index. The comparison tests whether
account identifiers reveal chain choices even when the marginal assignment frequency is uniform.

\textbf{Result.} For the specified \lit{csprng} sampler using \lit{os.urandom} and the seeded
reproducibility arm, T1--T3 produced no rejection at any tested depth, and T4 raised no
three-standard-error flag at any tested depth or adversary size. T5 has no rejection criterion.
Table~\ref{tab:enrollment_leakage} compares the specified sampler with five deliberate violations
of its assignment conditions at $m=10^5$.

\begin{table*}[!tp]
\centering
\caption{Top-$1$ prediction accuracy for a victim's chain by assignment scheme and depth at $m=10^5$ adversary enrollments.
Parentheses show accuracy divided by chance. \lit{csprng} is the specified sampler. Multipliers use
unrounded simulation values. The $k{=}5$ column shows five decimals so the multipliers follow
from the printed accuracies.}
\label{tab:enrollment_leakage}
\footnotesize
\begin{tabular}{@{}l ccccc@{}}
\toprule
\textbf{Scheme} & $k{=}1$ & $k{=}2$ & $k{=}3$ & $k{=}4$ & $k{=}5$ \\
\midrule
\lit{csprng} \textit{(spec)} & 0.1254 ($1\times$) & 0.0165 ($1\times$) & 0.0022 ($1\times$) & 0.0003 ($1\times$) & 0.00006 ($1\times$) \\
\lit{round\_robin} & \textbf{1.0000} ($8\times$) & \textbf{1.0000} ($58\times$) & \textbf{1.0000} ($422\times$) & \textbf{1.0000} ($3054\times$) & \textbf{1.00000} ($21822\times$) \\
\lit{id\_derived} & \textbf{1.0000} ($8\times$) & \textbf{1.0000} ($58\times$) & \textbf{1.0000} ($422\times$) & \textbf{1.0000} ($3054\times$) & \textbf{1.00000} ($21822\times$) \\
\lit{small\_pool} & 0.1259 ($1\times$) & 0.0621 ($4\times$) & 0.0624 ($26\times$) & 0.0623 ($190\times$) & 0.06238 ($1361\times$) \\
\lit{modbias8} & 0.1264 ($1\times$) & 0.0185 ($1\times$) & 0.0039 ($2\times$) & 0.0041 ($13\times$) & 0.00412 ($90\times$) \\
\lit{lcg} & 0.1250 ($1\times$) & 0.0173 ($1\times$) & 0.0025 ($1\times$) & 0.0003 ($1\times$) & 0.00010 ($2\times$) \\
\bottomrule
\end{tabular}
\end{table*}

\lit{round\_robin} and \lit{id\_derived} are exactly uniform over $\mathcal{V}_k$. T1 finds no
marginal nonuniformity for either at any depth, yet both are completely predictable. The \lit{freq}
predictor cannot distinguish \lit{id\_derived} from the specified sampler, while \lit{struct}
recovers its assignments. \lit{small\_pool} saturates at $1/16$, its
pool size, however large $\mathcal{V}_k$ becomes. The practical consequence is that a deployment
cannot validate its chain assignment by testing the output distribution for uniformity. The
requirements are that draws be CSPRNG-generated and independent, drawn from the full accepted set
$\mathcal{V}_k$ (chance repetition of a chain across accounts is expected and permitted, not a defect),
and not derivable from anything the adversary can observe, including account identifier, email address,
timestamp, or sequence position. Confidentiality of $K$ within the CR is an architectural property rather than a
sampler property and is not simulated here.

\section{Figures}\label{app:VII}
Figures~\ref{fig:S_naive_diagnostics}--\ref{fig:security_cost} provide the additional recovery and overhead diagnostics cited in \S~VII of the main paper.
\begingroup
\setlength{\intextsep}{6pt}
\begin{figure*}[!tbp]

  \centering
  \subfigure[Recovery vs.\ rule-file size.\label{fig:S_ruleset}]{%
  \begin{minipage}[t]{0.45\textwidth}
    \centering
    \includegraphics[width=\linewidth]{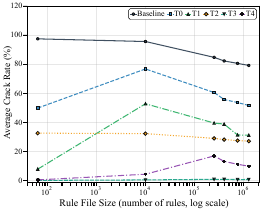}
  \end{minipage}}\hfill
  \subfigure[Hash algorithm.\label{fig:S_algo}]{%
  \begin{minipage}[t]{0.45\textwidth}
    \centering
    \includegraphics[width=\linewidth]{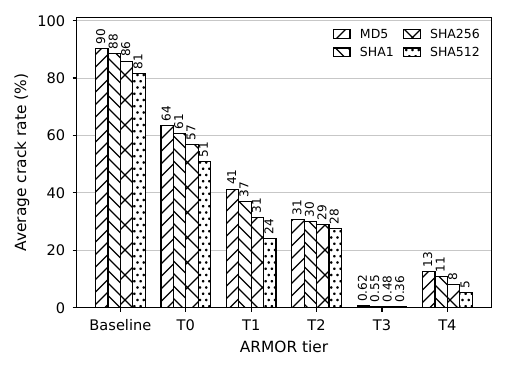}
  \end{minipage}}

  \caption{Additional \textsc{naive}-recovery diagnostics on A100. \textsc{naive} recovery does not scale monotonically with rule-file size, with T3 near zero across the $64$--$1.27$M-rule span, and hash-algorithm effects are secondary to the Baseline-to-T3 gap.}
  \label{fig:S_naive_diagnostics}
\end{figure*}
\endgroup

\begin{figure*}[!tbp]
  \centering
  \includegraphics[width=0.55\linewidth]{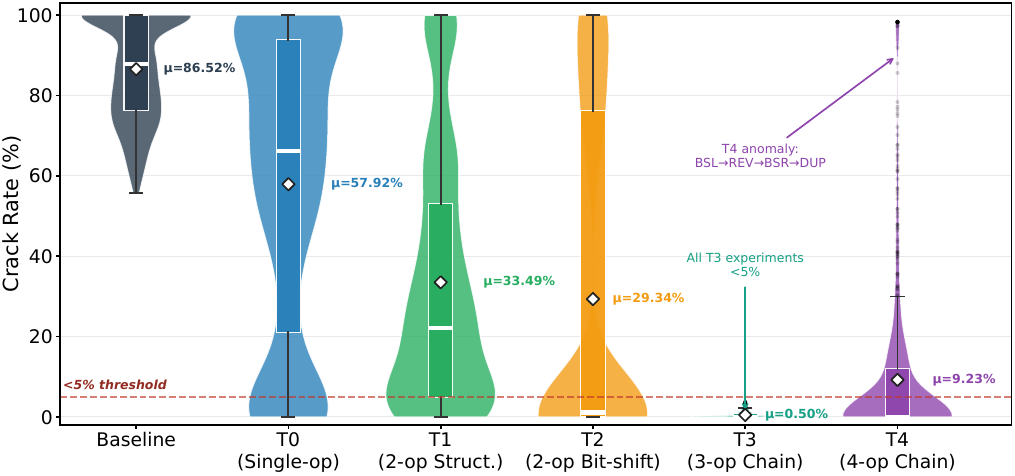}
  \caption{Distribution of rule-based \textsc{naive} crack rates by tier (A100). T3 forms a near-zero spike. The T4 bimodality reflects the validator-accepted \tablelit{T4-BC3} (\chain{BSL$\to$REV$\to$BSR$\to$DUP}) chain's non-adjacent algebraic degeneration.}
  \label{fig:S_violin}
\end{figure*}
\begin{figure*}[!tbp]
\setcounter{subfigure}{0}
  \centering
  \subfigure[Dataset-size sensitivity.\label{fig:S_size}]{%
  \begin{minipage}[t]{0.332\textwidth}
    \centering
    \includegraphics[width=\linewidth]{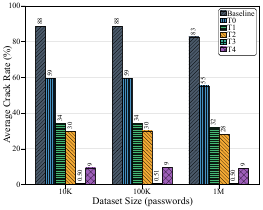}
  \end{minipage}}\hfill
  \subfigure[Duration sensitivity.\label{fig:S_duration}]{%
  \begin{minipage}[t]{0.332\textwidth}
    \centering
    \includegraphics[width=\linewidth]{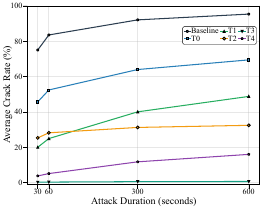}
  \end{minipage}}\hfill
  \subfigure[Hardened vs.\ budget.\label{fig:S_hardened_budget}]{%
  \begin{minipage}[t]{0.332\textwidth}
    \centering
    \includegraphics[width=\linewidth]{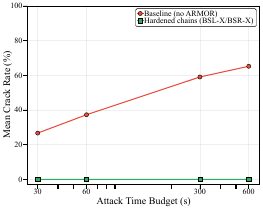}
  \end{minipage}}

  \caption{Crack-rate sensitivity diagnostics on A100. Baseline/T0 improve with longer runs while T3 changes negligibly. Recovery remains stable across two orders of magnitude in target-set size. Hardened recovery is stable across attack budgets.}
  \label{fig:S_sensitivity_diagnostics}
\end{figure*}
\begin{figure*}[!tbp]
\setcounter{subfigure}{0}
  \centering
  \subfigure[Pi~3 primitive latency.\label{fig:pi_latency}]{%
  \begin{minipage}[t]{0.32\textwidth}
    \centering
    \includegraphics[width=\linewidth]{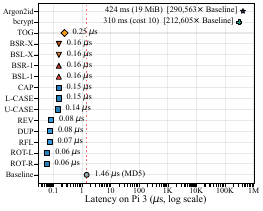}
  \end{minipage}}\hfill
  \subfigure[Latency vs.\ input length.\label{fig:latency_length_panel}]{%
  \begin{minipage}[t]{0.32\textwidth}
    \centering
    \includegraphics[width=\linewidth]{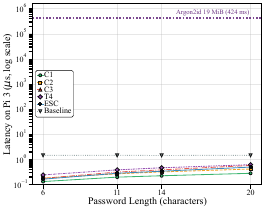}
  \end{minipage}}\hfill
  \subfigure[Latency: ARMOR vs.\ KDFs.\label{fig:security_hierarchy}]{%
  \begin{minipage}[t]{0.32\textwidth}
    \centering
    \includegraphics[width=\linewidth]{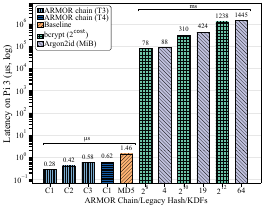}
  \end{minipage}}
  \caption{Defender-side IoT/edge overhead. ARMOR remains microsecond-scale on the Pi~3 reference and stays far below memory-hard KDF costs.}
  \label{fig:security_cost}
\end{figure*}

\section{Bounded Production Attack and Login Recovery}\label{app:VIII}
\label{supp:bounded_production}
This experiment uses the specified production construction, assigned accepted $k{=}3$ chains with fixed terminal \chain{ESC}; it is separate from the $14$-operation historical coverage experiment. The attacker receives salted verifiers, public primitives and the bounded candidate corpus, but not the victim's chain or the grader's ground truth. The first $10^6$ eligible $6$--$25$-byte RockYou entries contain $999{,}997$ unique candidates. We use the first $1{,}000$ retained targets of a $10{,}000$-target set; all passwords occur in this pool. The historical hash-bound generator uses SHA-256 counter/rejection sampling of eligible entries without replacement, seed \lit{ARMOR-W2-production-targets-v1}, and nested target prefixes. Domain-separated functions derive each assigned chain index and $16$-byte salt from that seed and account ordinal. The bounded prior is nearly flat: $999{,}994$ candidates occur once; \lit{daniel}, \lit{te amo}, and \lit{i love u} each occur twice. For ARMOR, enumeration of $999{,}997\times422=421{,}998{,}734$ candidate-chain pairs yields $132{,}773{,}053$ distinct transformed values, of maximum length $400$ bytes. Plain candidates are ordered by descending frequency, then first occurrence. The two arms use paired passwords and the same $16$-byte salts.

\textbf{Schedule, batching, and correctness.} Transformed hypotheses are ordered by descending aggregate frequency times chain multiplicity, with lexicographic byte ties, under the empirical corpus prior and an independent uniform chain-assignment model. The measured work ratio applies to this implemented order and single target set. This is a unit-cost ordering, not a GPU-cost-aware optimal attack. CUDA replays the CPU-built order; the reference and optimized engines share that order and the same adaptive target-removal policy. Length-appropriate short and general MD5 paths process the actual bytes. For $R$ schedule records, batch size is $\max(1,\min(\lfloor R/200\rfloor,2^{32}-1))$. ARMOR uses $663{,}865$ hypotheses per batch and stops after $200$ batches ($132{,}773{,}000$ hypotheses); plain uses $4{,}999$ and $201$ batches including the final partial batch. All targets are recovered before ARMOR's final $53$ untested hypotheses; this is full target recovery, not full schedule exhaustion.

Every batch hashes against its start-of-batch active-target snapshot. If $A_j$ accounts are active and $H_j$ hypotheses are executed, work is $H_j A_j$, including tests performed after a target's first hit inside that batch. Atomic minimum ranks identify the first matching hypothesis. All account ranks, batch work, checkpoint first crossings and overshoot counts reconcile with the frozen CPU-verified reference; reference and optimized details are byte-identical on both GPUs. Device self-tests include byte-domain transforms and MD5 length boundaries. Each suite performs one warm-up and three measured repetitions; all measured repeats, including the slower first plain repeat on A100, are retained.

\begin{table}[!ht]
\centering\footnotesize
\caption{Checkpoint work and actual target recovery, identical on both GPUs and both kernels. Nominal checkpoints can overshoot.}
\label{tab:S_bounded_checkpoints}
\setlength{\tabcolsep}{3pt}
\begin{tabular}{@{}rrrrr@{}}
\toprule
Nominal & Recovered A/P & ARMOR checks & Plain checks & Ratio\\
\midrule
$1\%$ & $36/15$ & $663{,}865{,}000$ & $14{,}922{,}015$ & $44.49$\\
$10\%$ & $102/102$ & $4{,}378{,}853{,}540$ & $99{,}865{,}023$ & $43.85$\\
$50\%$ & $503/501$ & $27{,}146{,}103{,}715$ & $401{,}549{,}674$ & $67.60$\\
$90\%$ & $902/900$ & $44{,}623{,}013{,}705$ & $511{,}417{,}696$ & $87.25$\\
$100\%$ & $1000/1000$ & $46{,}152{,}558{,}665$ & $516{,}541{,}868$ & $89.35$\\
\bottomrule\end{tabular}
\end{table}
The nominal $1\%$ checkpoints are actual $3.6\%$ ARMOR and $1.5\%$ plain recovery and cannot be interpreted as a matched $1\%$-recovery cost ratio. All checkpoint ratios refer to the observed fractions in Table~\ref{tab:S_bounded_checkpoints}.

\textbf{Hardware and software.} A6000: NVIDIA RTX A6000, $48$\,GB, compute capability $8.6$, AMD Ryzen Threadripper 9960X host, driver $595.99.02$. A100: NVIDIA A100 $80$\,GB PCIe, compute capability $8.0$, AMD EPYC 7702 host, driver $595.71.05$. Both use Ubuntu 24.04.4, nvcc 12.0.140, and Hashcat 6.2.6. The GPU and host changes are recorded together; whole-pipeline differences cannot be attributed solely to GPU architecture.

\begin{table*}[!ht]
\centering\footnotesize
\caption{Full-recovery CUDA times in seconds, mean $\pm$ sample SD of three repeats. A/P ratios are ratios of means. ``Wall'' means the search loop, not the full attack.}
\label{tab:S_bounded_timings}
\setlength{\tabcolsep}{4pt}
\begin{tabular}{@{}llrrrrrr@{}}
\toprule
GPU & Engine & A kernel & P kernel & Kernel A/P & A wall & P wall & Wall A/P\\
\midrule
A6000 & Reference & $9.345592\pm0.028420$ & $0.081816\pm0.000541$ & $114.23$ & $9.349762\pm0.028486$ & $0.085902\pm0.000139$ & $108.84$\\
A6000 & Optimized & $9.228750\pm0.047172$ & $0.072967\pm0.000258$ & $126.48$ & $9.232960\pm0.047158$ & $0.077013\pm0.000347$ & $119.89$\\
A100 & Reference & $8.121694\pm0.000127$ & $0.073956\pm0.000808$ & $109.82$ & $8.132312\pm0.000080$ & $0.081686\pm0.001472$ & $99.56$\\
A100 & Optimized & $7.356568\pm0.000333$ & $0.061021\pm0.001759$ & $120.56$ & $7.369107\pm0.000682$ & $0.069140\pm0.002565$ & $106.58$\\
\bottomrule\end{tabular}\end{table*}

\textbf{Timing interpretation and preparation.} CUDA events sum kernel execution; wall time covers the recovery loop. CPU group-by and sort cost $50.6+55.9=106.5$\,s on the A6000 host and $97.1+136.6=233.7$\,s on the A100 host, excluding subsequent export overhead. These costs are excluded from Table~\ref{tab:S_bounded_timings}; the corpus-and-chain schedule can be reused for additional account sets under the same model, but such amortization must be explicit. Transfers are reported separately and partially overlap search-wall operations, so summing those counters is not an end-to-end measurement. Reference-to-optimized kernel speedups are $1.0127$ ARMOR/$1.1213$ plain on A6000 and $1.1040$/$1.2120$ on A100. The improvement depends on both input length and hardware. Repeat SD is within-run timing variation, not cross-sample uncertainty.

\textbf{External plain-only anchor.} Hashcat mode $10$ with optimized kernels and hexadecimal candidates correctly recovers all $1{,}000$ original plain passwords on both GPUs in the warm-up and each repeat. Its observed maximum candidate length is $31$ bytes for the optimized kernel and $256$ for the pure kernel; neither covers the full $400$-byte ARMOR schedule. The plain candidates are $6$--$25$ bytes. Hashcat measured full-process wall times are $1.452$, $2.454$, $1.456$\,s on A6000 and $6.645$, $6.634$, $6.681$\,s on A100. Its terminal progress is approximately $999$ million checks, versus $516.542$ million custom plain checks after adaptive target removal. Work, batching, startup and timing boundaries differ. This is an external-tool correctness and process-time anchor, not comparable-work optimized-throughput calibration. Terminal speed readings are not sustained throughput. No fastest-attacker or energy/cost claim follows.

\textbf{Complete ambiguity sets and posterior scoring.} For each recovered transformed value $z^*$, Stage~2 enumerates all $(p,K)$ pairs in the bounded pool satisfying $\mathrm{ESC}(T_K(p))=z^*$. Inverse enumeration attempts $1{,}000\times422=422{,}000$ chain inversions and checks membership in the same corpus, with a single corpus-frequency pass and forward verification of all $4{,}662$ surviving pairs. Brute-force enumeration independently produces the same sets. Candidate ranking uses only attacker-available data:
\[
\Pr(p\mid z^*)\ \propto\ f(p)\,\bigl|\{K\in\mathcal{V}_3:\mathrm{ESC}(T_K(p))=z^*\}\bigr|,
\]
where $f(p)$ is the bounded-corpus occurrence count; ties are broken by plaintext bytes. Ground truth is withheld from enumeration and scoring and used only for evaluation. This empirical posterior assumes the stated bounded password prior and independent uniform chains, not a universal password distribution.

\textbf{Simulated login-valid recovery and distinct cost axes.} The separate grader walks posterior-ordered plaintext guesses, transforms each under the account's actual assigned chain, and checks its actual salted verifier. All $1{,}000$ original passwords and login-valid passwords are recovered; $908$ sets contain only one distinct plaintext. Realized submissions to first login-valid success have mean $1.086$, median $1$, maximum $6$. The inverse Stage~2 pipeline takes $1.196958$\,s, including loading, inversions, corpus-frequency counting, forward checks and output; its brute-force reference takes approximately $24.4$\,s and is correctness corroboration, not the cheapest attack estimate. These are the retained October~10 same-target CPU results, reused by identical targets and recovered ranks/bytes in the GPU suites. They were measured on the Threadripper 9960X CPU, not on the A100 host. They are not new GPU-phase measurements. Simulated submission counts are distinct from offline compute seconds, and no production login latency, throttling, lockout or network cost was measured.

\textbf{Limits and evidence.} All targets are in-pool, so complete recovery and singleton certainty are conditional on the bounded pool. The salted-MD5 experiment does not establish a multiplier for other hashes/KDFs or broader corpora. It does not implement a GPU-cost-aware optimal schedule, a matched externally optimized full attack, or an isolated-pepper attack comparison. A high-entropy isolated pepper remains cryptographically stronger under DB-only compromise. Full authentication, coordinated migration, registry recovery and energy/monetary cost remain outside these measurements. Immutable source, input, export and binary hashes, raw batch/timing CSVs, Hashcat outputs accompany the retained evidence and are available on request. Archived records identify the original schedule binaries by hashes and byte sizes; those binaries are absent from the local evidence copy.

\raggedbottom

\FloatBarrier

\renewcommand{\IEEEbibitemsep}{-0.5pt plus .1pt minus .1pt}
\bibliographystyle{IEEEtran}  
\bibliography{main} 
\end{document}